\documentclass[aps,prresearch,reprint,superscriptaddress]{revtex4-2}
\usepackage{amsmath,amsfonts,amssymb}
\usepackage{mathtools}
\usepackage{multirow}
\usepackage{ulem}
\usepackage{algorithmic}
\usepackage[ruled, vlined, linesnumbered]{algorithm2e}

\usepackage[utf8]{inputenc}

\newcommand{\ud}[1]{\textcolor{black}{#1}}
\usepackage[T1]{fontenc}

\usepackage{bm}
\usepackage{times}

\usepackage{mathtools}
\usepackage{amsmath}
\usepackage[shortlabels]{enumitem}

\usepackage{graphicx,epic,eepic,epsfig,amsmath,latexsym,amssymb,verbatim,color}
 
\usepackage{amsfonts}       
\usepackage{nicefrac}       

\usepackage{amsmath}
\usepackage{bbm}

\usepackage{float}
\usepackage{tikz}
\usetikzlibrary{chains}
\usetikzlibrary{fit}
\usetikzlibrary{arrows}

\usepackage{epsfig}
\usetikzlibrary{shapes.symbols,patterns} 
\usepackage{pgfplots}
\pgfplotsset{compat=1.18}

\usepackage[strict]{changepage}
\usepackage{hyperref}
\hypersetup{colorlinks=true,citecolor=blue,linkcolor=blue,filecolor=blue,urlcolor=blue,breaklinks=true}

\usepackage[marginal]{footmisc}
\usepackage{url}
\usepackage{theorem}

\newtheorem{proposition}{Proposition}
\newtheorem{lemma}[proposition]{Lemma}

\def\squareforqed{\hbox{\rlap{$\sqcap$}$\sqcup$}}
\def\qed{\ifmmode\squareforqed\else{\unskip\nobreak\hfil
\penalty50\hskip1em\null\nobreak\hfil\squareforqed
\parfillskip=0pt\finalhyphendemerits=0\endgraf}\fi}
\def\endenv{\ifmmode\;\else{\unskip\nobreak\hfil
\penalty50\hskip1em\null\nobreak\hfil\;
\parfillskip=0pt\finalhyphendemerits=0\endgraf}\fi}
\newenvironment{proof}{\noindent \textbf{{Proof~} }}{\hfill $\blacksquare$}

\newcounter{remark}

\newcounter{example}

\mathchardef\ordinarycolon\mathcode`\:
\mathcode`\:=\string"8000
\def\vcentcolon{\mathrel{\mathop\ordinarycolon}}
\begingroup \catcode`\:=\active
  \lowercase{\endgroup
  \let :\vcentcolon
  }

\usepackage{cleveref}
\usepackage{graphicx}
\usepackage{xcolor}

\RequirePackage[framemethod=default]{mdframed}
\newmdenv[skipabove=7pt,
skipbelow=7pt,
backgroundcolor=darkblue!15,
innerleftmargin=5pt,
innerrightmargin=5pt,
innertopmargin=5pt,
leftmargin=0cm,
rightmargin=0cm,
innerbottommargin=5pt,
linewidth=1pt]{tBox}

\newmdenv[skipabove=7pt,
skipbelow=7pt,
backgroundcolor=blue2!25,
innerleftmargin=5pt,
innerrightmargin=5pt,
innertopmargin=5pt,
leftmargin=0cm,
rightmargin=0cm,
innerbottommargin=5pt,
linewidth=1pt]{dBox}
\newmdenv[skipabove=7pt,
skipbelow=7pt,
backgroundcolor=darkkblue!15,
innerleftmargin=5pt,
innerrightmargin=5pt,
innertopmargin=5pt,
leftmargin=0cm,
rightmargin=0cm,
innerbottommargin=5pt,
linewidth=1pt]{sBox}
\definecolor{darkblue}{RGB}{0,76,156}
\definecolor{darkkblue}{RGB}{0,0,153}
\definecolor{blue2}{RGB}{102,178,255}
\definecolor{darkred}{RGB}{195,0,0}

\newcommand{\nc}{\newcommand}
\nc{\rnc}{\renewcommand}
\nc{\lbar}[1]{\overline{#1}}
\nc{\bra}[1]{\langle#1|}
\nc{\ket}[1]{|#1\rangle}
\nc{\ketbra}[2]{|#1\rangle\!\langle#2|}
\nc{\braket}[2]{\langle#1|#2\rangle}

\nc{\proj}[1]{| #1\rangle\!\langle #1 |}
\nc{\avg}[1]{\langle#1\rangle}
\nc{\rank}{\operatorname{Rank}}
\nc{\smfrac}[2]{\mbox{$\frac{#1}{#2}$}}
\nc{\tr}{\operatorname{Tr}}
\nc{\ox}{\otimes}
\nc{\dg}{\dagger}
\nc{\dn}{\downarrow}
\nc{\cA}{{\cal A}}
\nc{\cB}{{\cal B}}
\nc{\cC}{{\cal C}}
\nc{\cD}{{\cal D}}
\nc{\cE}{{\cal E}}
\nc{\cF}{{\cal F}}
\nc{\cG}{{\cal G}}
\nc{\cH}{{\cal H}}
\nc{\cI}{{\cal I}}
\nc{\cJ}{{\cal J}}
\nc{\cK}{{\cal K}}
\nc{\cL}{{\cal L}}
\nc{\cM}{{\cal M}}
\nc{\cN}{{\cal N}}
\nc{\cO}{{\cal O}}
\nc{\cP}{{\cal P}}
\nc{\cQ}{{\cal Q}}
\nc{\cR}{{\cal R}}
\nc{\cS}{{\cal S}}
\nc{\cT}{{\cal T}}
\nc{\cU}{{\cal U}}
\nc{\cV}{{\cal V}}
\nc{\cX}{{\cal X}}
\nc{\cY}{{\cal Y}}
\nc{\cZ}{{\cal Z}}
\nc{\cW}{{\cal W}}
\nc{\csupp}{{\operatorname{csupp}}}
\nc{\qsupp}{{\operatorname{qsupp}}}
\nc{\var}{{\operatorname{var}}}
\nc{\rar}{\rightarrow}
\nc{\lrar}{\longrightarrow}
\nc{\polylog}{{\operatorname{polylog}}}
\nc{\wt}{{\operatorname{wt}}}
\nc{\av}[1]{{\left\langle {#1} \right\rangle}}
\nc{\supp}{{\operatorname{supp}}}

\nc{\argmin}{{\operatorname{argmin}}}

\def\i{\mathbf{i}}

\def\x{\xi}

\nc{\RR}{{{\mathbb R}}}
\nc{\CC}{{{\mathbb C}}}
\nc{\FF}{{{\mathbb F}}}
\nc{\NN}{{{\mathbb N}}}
\nc{\ZZ}{{{\mathbb Z}}}
\nc{\PP}{{{\mathbb P}}}
\nc{\QQ}{{{\mathbb Q}}}
\nc{\UU}{{{\mathbb U}}}
\nc{\EE}{{{\mathbb E}}}
\nc{\id}{{\operatorname{id}}}

\nc{\CHSH}{{\operatorname{CHSH}}}

\nc{\be}{\begin{equation}}
\nc{\ee}{{\end{equation}}}
\nc{\bea}{\begin{eqnarray}}
\nc{\eea}{\end{eqnarray}}
\nc{\<}{\langle}
\rnc{\>}{\rangle}
\nc{\rU}{\mbox{U}}

\nc{\ob}[1]{#1}

\nc{\SEP}{{\text{\rm SEP}}}
\nc{\NS}{{\text{\rm NS}}}
\nc{\LOCC}{{\text{\rm LOCC}}}
\nc{\PPT}{{\text{\rm PPT}}}
\nc{\EXT}{{\text{\rm EXT}}}
\nc{\Sym}{{\operatorname{Sym}}}

\nc{\ERLO}{{E_{\text{r,LO}}}}
\nc{\ERLOCC}{{E_{\text{r,LOCC}}}}
\nc{\ERPPT}{{E_{\text{r,PPT}}}}
\nc{\ERLOCCinfty}{{E^{\infty}_{\text{r,LOCC}}}}
\nc{\Aram}{{\operatorname{\sf A}}}

\usepackage{tikz}
\usepackage{hyperref}
\hypersetup{colorlinks=true,citecolor=blue,linkcolor=blue,filecolor=blue,urlcolor=blue,breaklinks=true}

\makeatletter
\def\grd@save@target#1{%
  \def\grd@target{#1}}
\def\grd@save@start#1{%
  \def\grd@start{#1}}
\tikzset{
  grid with coordinates/.style={
    to path={%
      \pgfextra{%
        \edef\grd@@target{(\tikztotarget)}%
        \tikz@scan@one@point\grd@save@target\grd@@target\relax
        \edef\grd@@start{(\tikztostart)}%
        \tikz@scan@one@point\grd@save@start\grd@@start\relax
        \draw[minor help lines,magenta] (\tikztostart) grid (\tikztotarget);
        \draw[major help lines] (\tikztostart) grid (\tikztotarget);
        \grd@start
        \pgfmathsetmacro{\grd@xa}{\the\pgf@x/1cm}
        \pgfmathsetmacro{\grd@ya}{\the\pgf@y/1cm}
        \grd@target
        \pgfmathsetmacro{\grd@xb}{\the\pgf@x/1cm}
        \pgfmathsetmacro{\grd@yb}{\the\pgf@y/1cm}
        \pgfmathsetmacro{\grd@xc}{\grd@xa + \pgfkeysvalueof{/tikz/grid with coordinates/major step}}
        \pgfmathsetmacro{\grd@yc}{\grd@ya + \pgfkeysvalueof{/tikz/grid with coordinates/major step}}
        \foreach \x in {\grd@xa,\grd@xc,...,\grd@xb}
        \node[anchor=north] at (\x,\grd@ya) {\pgfmathprintnumber{\x}};
        \foreach \y in {\grd@ya,\grd@yc,...,\grd@yb}
        \node[anchor=east] at (\grd@xa,\y) {\pgfmathprintnumber{\y}};
      }
    }
  },
  minor help lines/.style={
    help lines,
    step=\pgfkeysvalueof{/tikz/grid with coordinates/minor step}
  },
  major help lines/.style={
    help lines,
    line width=\pgfkeysvalueof{/tikz/grid with coordinates/major line width},
    step=\pgfkeysvalueof{/tikz/grid with coordinates/major step}
  },
  grid with coordinates/.cd,
  minor step/.initial=.2,
  major step/.initial=1,
  major line width/.initial=2pt,
}
\makeatother

\usepackage{thmtools}
\usepackage{thm-restate}
\usepackage{etoolbox}
\makeatletter
\def\problem@s{}
\newcounter{problems@cnt}

\newcommand{\allproblems}{\problem@s}
\makeatother

\definecolor{beamer}{rgb}{0.2,0.2,0.7}
\definecolor{colorone}{rgb}{1,0.36,0.03}
\definecolor{colortwo}{rgb}{0.4,0.77,0.17}
\definecolor{colorthree}{rgb}{0.01,0.51,0.93}
\definecolor{colorfour}{rgb}{0.47,0.26,0.58}
\definecolor{colorfive}{rgb}{0.12,0.55,0.16}
\usepackage{tcolorbox}
\usepackage{relsize}
\usepackage{graphicx}
\usepackage{booktabs}
\usepackage{amsmath}
\usepackage{tikz}
\usetikzlibrary{quantikz}
\nc{\st}{\text{subject to} \ }
\nc{\supre}{\text{supremum} \ }
\nc{\sdp}{\text{sdp}}
\usepackage{array}

\allowdisplaybreaks

\nc{\ith}[1]{{#1}^\mathrm{th}}

\begin{document}
\title{Dynamic LOCC Circuits for Automated Entanglement Manipulation}

\author{Xia Liu}
\affiliation{Thrust of Artificial Intelligence, Information Hub, The Hong Kong University of Science and Technology (Guangzhou), Guangdong 511453, China}

\author{Jiayi Zhao}
\affiliation{Thrust of Artificial Intelligence, Information Hub, The Hong Kong University of Science and Technology (Guangzhou), Guangdong 511453, China}

\author{Benchi Zhao}
\email{benchizhao@gmail.com}
\affiliation{QICI Quantum Information and Computation Initiative, Department of Computer Science, The University of Hong Kong, Pokfulam Road, Hong Kong}

\author{Xin Wang}
\email{felixxinwang@hkust-gz.edu.cn}
\affiliation{Thrust of Artificial Intelligence, Information Hub, The Hong Kong University of Science and Technology (Guangzhou), Guangdong 511453, China}

\begin{abstract}
Due to the limited qubit number of quantum devices, distributed quantum computing is considered a promising pathway to overcome this constraint. In this paradigm, multiple quantum processors are interconnected to form a cohesive computational network, and the most natural set of free operations is local operations and classical communication (LOCC). However, designing a practical LOCC protocol for a particular task has been a tough problem. In this work, we propose a general and flexible framework called dynamic LOCCNet (DLOCCNet) to simulate and design LOCC protocols. We demonstrate its effectiveness in two key applications: entanglement distillation and distributed state discrimination. The protocols designed by DLOCCNet, in contrast to conventional ones, can solve larger-sized problems with reduced training time, making the framework a practical and scalable tool for current quantum devices. This work advances our understanding of the capabilities and limitations of LOCC while providing a powerful methodology for protocol design.

\end{abstract}

\maketitle

\section{Introduction}
Quantum computing holds transformative potential over classical computing by exploiting quantum superposition and entanglement, enabling speedups in solving complex problems such as optimization~\cite{farhi2014quantum,harrigan2021quantum}, cryptography~\cite{bennett2014quantum, ekert1991quantum}, and quantum simulations~\cite{georgescu2014quantum, daley2022practical, buluta2009quantum}. However, current quantum processors remain constrained by limited qubit counts. To overcome these limitations, distributed quantum computing~\cite{caleffi2024distributed,peng2020simulating,mitarai2021constructing,piveteau2023circuit} has emerged as a promising pathway, where multiple quantum processors are interconnected to form a cohesive computational quantum network. 

In the paradigm of distributed quantum computing, local operations and classical communication (LOCC) operations~\cite{chitambar2014everything} are the free operations. To connect different parties and realize computational tasks, shared Bell states are necessary~\cite{bennett1993teleporting, bouwmeester1997experimental,huang2004experimental,chen2024optimum}. Since quantum systems are vulnerable and easily corrupted by noise, distributing high-fidelity Bell states is very expensive, and we must assume the distributed entanglement to be limited. This constraint strongly motivates us to design LOCC protocols that can utilize distributed entanglement in the most efficient way.


Although the basic idea of LOCC is relatively easy to grasp, its mathematical structure is highly complicated and difficult to characterize~\cite{chitambar2014everything}, implying that we can hardly design the optimal LOCC protocol to realize our tasks. In previous work, Zhao et al.~\cite{zhao2021practical} proposed an optimization framework called LOCCNet to simulate and design LOCC protocols via optimizing parameterized LOCC protocols. With this framework, they found the state-of-the-art LOCC protocol to purify Bell states from mixed states and isotropic states. However, such a framework requires increasingly high training costs as the size of the quantum system increases, and thus it is mainly limited to small-size regimes of entanglement manipulation. Training large-size LOCCNet may encounter barren plateaus~\cite{mcclean2018barren,Cerezo2023does,Zhang2024,zhang2020toward,Liu2024mitigating,zhang2024predicting}, which refers to the phenomenon where the gradient drops exponentially with respect to the size of the quantum systems. How to automatically design efficient LOCC protocols for large-scale entanglement manipulation remains a key challenge at the center of distributed quantum information processing.


\begin{figure}[t]
    \centering
    \includegraphics[width=\linewidth]{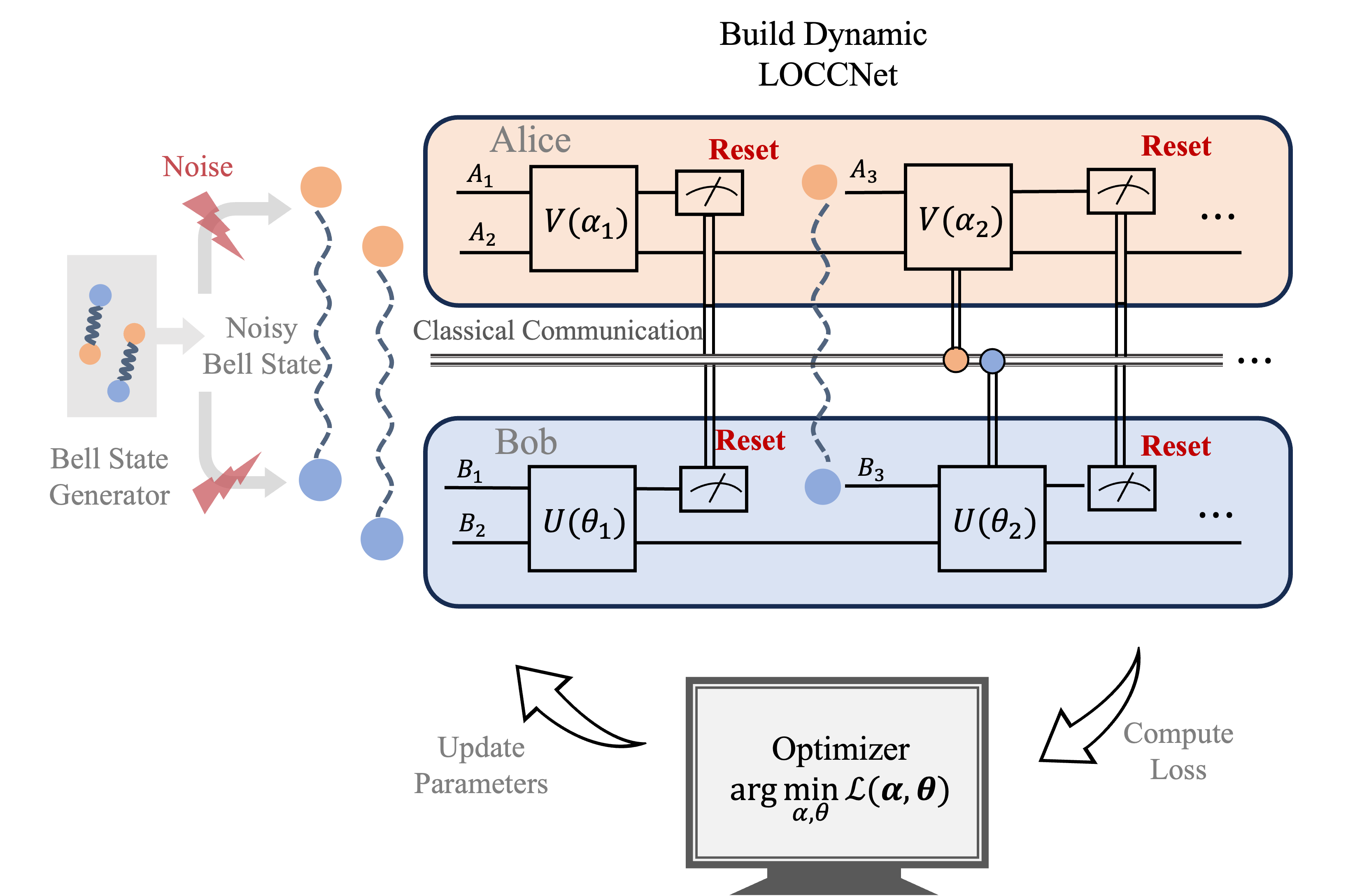}
    \caption{\textbf{Illustration of the procedure for optimizing an LOCC protocol with dynamic LOCCNet.} For simplicity, this diagram involves only two parties: Alice and Bob. Each local operation is encoded as the parameterized quantum circuits (PQCs), while the double line represents the classical communication channel between the parties. Following each measurement, the parties can reset their qubits and introduce fresh entangled states for subsequent rounds of PQC evolution. After the final round, a task-specific loss function, $\mathcal{L}(\mathbf{\alpha,\mathbf{\theta}})$, guides the training process. The divergent parameters across different PQCs indicate distinct possible measurement outcomes. Finally, optimization methods iteratively update the parameters, $\alpha_i,\theta_i$, in each local operation, yielding the optimized dynamic LOCC protocol.}
    \label{fig:dyn-loccnet}
\end{figure}

In this work, we propose a more general and flexible framework called dynamic LOCCNet (DLOCCNet), as shown in Fig.~\ref{fig:dyn-loccnet}, then apply this framework to design protocols for entanglement distillation~\cite{bennett1996purification, deutsch1996quantum,murao1998multiparticle,dur2007entanglement,pan2003experimental,devetak2005distillation,Wang2019b,Leditzky2017,Fang2017, jansen2022enumerating,Wang2016c} and distributed state discrimination~\cite{barnett2009quantum, bae2015quantum,zhu2025entanglement}. 
\ud{Specifically, in the task of entanglement distillation, we consider three kinds of noisy states: Bell states affected by erasure, depolarizing, and amplitude channels, respectively. We design a scalable entanglement distillation protocol for the Bell states affected by the erasure channel, which outperforms the conventional protocol. For the other two channels, the protocol designed within our framework could achieve fidelity higher than previous protocols.} \ud{Due to the exponential training cost required by LOCCNet, it is practically limited to designing entanglement distillation protocols for only a small number of noisy state copies. In contrast, our DLOCCNet framework overcomes this scalability limitation, enabling efficient optimization across arbitrary copy numbers.} In the task of distributed state discrimination, we study the average success probability of distinguishing two quantum states using different copy numbers. The result  demonstrates a substantial improvement in success probability as the number of state copies increases. Compared with previous work~\cite{zhao2021practical}, the newly proposed DLOCCNet framework offers several key advantages: $^{(1)}$ Our framework could solve problems involving more systems than the conventional one. $^{(2)}$ The training time of DLOCCNet is much shorter than the previous framework. Based on these advantages, our framework offers a more powerful and efficient method to design practical LOCC protocols for distributed quantum computing.

\section{Main Results}

\subsection{Framework of dynamic LOCCNet}

Here, we propose a more general and flexible framework called DLOCCNet, whose schematic is illustrated in Fig.~\ref{fig:dyn-loccnet}. \ud{In LOCC, it is known that the choice of local operation depends on the previous measurement results. This framework can not only simulate the implementation of quantum circuits, but also simulate the effects of previous measurements.}
Specifically, two remote labs, namely Alice and Bob, initially share an entangled state. They then apply local parametrized quantum circuits (PQC) $V(\alpha_1)$ and $U(\theta_1)$ to their respective systems. Afterwards, they perform measurements on parts of their systems, obtaining outcomes $m_A^{(1)}$ and $m_B^{(1)}$, which are communicated with each other through a classical channel. Next, Alice and Bob reset the measured system with a fresh shared entangled state and simulate PQCs $V(\alpha_2)$ and $U(\theta_2)$ on their own systems, which are determined by the previous measurement results $m_A^{(1)}, m_B^{(1)}$. This process of measurement, communication, reset, and simulation of PQCs is repeated iteratively. In the end, the loss function can be computed. A classical optimizer is used to minimize the cost function, update the PQC parameters, and restart the circuit. Through iterating the processes, the LOCC protocol is considered optimized when the loss function converges to its minimum value. \ud{In the end, the framework outputs a practical distributed dynamic protocol with specific quantum circuits, and the corresponding conditioning.}

In the following sections, we are going to apply the DLOCCNet to the tasks of entanglement distillation and state discrimination, and we also compare the results achieved by the conventional framework LOCCNet~\cite{zhao2021practical}. Note that the numerical simulations are conducted on QuAIRKit~\cite{quairkit}.

\subsection{Entanglement distillation}

\begin{figure*}[t]
    \centering
    \includegraphics[width=0.8\linewidth]{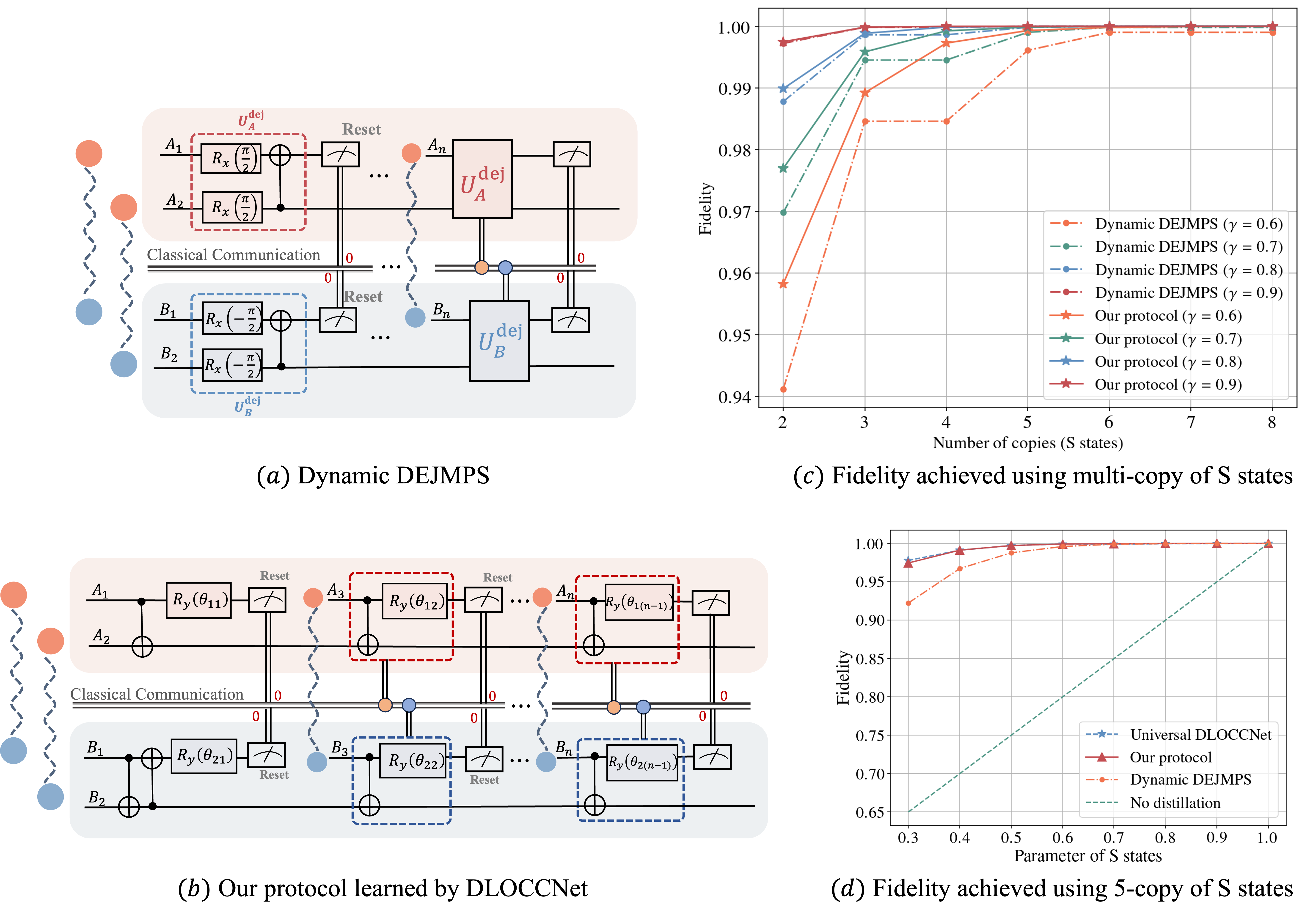}
    \caption{\textbf{Fidelity achieved by distillation methods for maximally entangled states affected by bilocal erasure channel.} $(a)$ is the dynamic DEJMPS for $n$ copies of an S state, where $U_A^{\rm dej}, U_B^{\rm dej}$ compose the DEJMPS protocol. $(b)$ shows the simplified circuit learned by dynamic LOCCNet, where the rotation angles of $R_y$ gates are $\theta_{11}=\theta_{21}=\arccos{(1-\gamma)+\pi}$, and $\theta_{ij}=-\frac{\pi}{2}$ for $i\in\{1,2\}, j\in\{2,\cdots,n-1\}$. $(c)$ is the fidelity achieved using two kinds of S state distillation protocol, where the horizontal axis represents the number of copies, and the lines of the same color represent the fidelity achieved under the same noise parameters. $(d)$ shows the fidelity achieved with the noise parameter using $5$-copy of S states.}
    \label{fig:s-n}
\end{figure*}

Given the scalability constraints of existing quantum hardware, distributed quantum computing offers a more practical alternative to solve the large-scale problem. In this paradigm, only LOCC operations are allowed, which limits the power of quantum computing. To achieve universal computation on distributed quantum computers, presharing maximally entangled states, generally Bell state $\ket{\Phi^+} = 1/\sqrt{2} (\ket{00}+\ket{11})$, is considered a promising method. While perfectly preshared Bell states enable accurate quantum operations, realistic scenarios inevitably involve noise corruption during qubit transmission and storage. This degradation significantly impacts the reliability of distributed quantum protocols.

To address this fundamental challenge, entanglement distillation protocols~\cite{bennett1996purification,deutsch1996quantum, jansen2022enumerating} were developed to generate high-fidelity Bell states by consuming multiple copies of noisy entangled states. These protocols can be applied recursively, consuming additional noisy copies to achieve progressively higher fidelity. Besides, these distillation protocols can not only provide high-fidelity Bell states, but also enhance the state estimation~\cite{casapao2024disti}

In the following, we apply our framework, DLOCCNet, to the design of entanglement distillation protocols. Specifically, for the distillation of $n$ copies of the bipartite noisy Bell state $\rho_{AB}=\cN(\proj{\Phi^+})$, where $\cN$ is the noisy channel, we construct $2k$ parameterized quantum circuits (PQCs). These consist of $k$ circuits for Alice's operations and $k$ for Bob's, with an integer number $k\in[1,n-1]$ determined by the classical communication rounds. Upon successful completion, these PQCs are designed to output a state with enhanced fidelity relative to the Bell state.
The optimization process is formulated by a cost function that minimizes infidelity $\mathcal{L} = 1-F$, where fidelity $F=\langle\Phi^+|\rho_{AB}|\Phi^+\rangle$ quantifies the overlap between the output state and the target Bell state $|\Phi^+\rangle$. We study the distillation protocols for the Bell states corrupted by different noisy models, including the erasure channel $\cN_{\rm era}$, the depolarizing channel $\cN_{\rm dep}$, and the amplitude damping channel $\cN_{\rm ad}$.

The erasure channel~\cite{grassl1997codes,bennett1997capacities} is a fundamental noise model in quantum information theory that characterizes the probabilistic loss or erasure of a quantum state during transmission. It is mathematically described by
\begin{align}
    \cN_{\rm era}(\rho)=\gamma\rho+(1-\gamma)\ketbra{0}{0},
\end{align}
where $1-\gamma$ is the erasure probability, and $\ket{0}$ represents an erasure symbol that indicates the loss of the original quantum information. 
Consider a third-party Bell state generator that prepares and transmits the bell state to Alice and Bob through a bilocal noisy channel, as illustrated in Fig.~\ref{fig:dyn-loccnet}. If the noisy channel is characterized as an erasure channel, the resulting noise state is referred to as the S state,
\begin{equation}
    \rho_{\rm s}=\gamma\ketbra{\Phi^+}{\Phi^+}+(1-\gamma)\ketbra{00}{00}.
\end{equation}
A celebrated distillation protocol, DEJMPS~\cite{deutsch1996quantum}, performs on two copies of some S state and outputs a state whose fidelity to $\ket{\Phi^+}$ is $\frac{(1+p)^2}{2(1+p^2)}$. Now we present a protocol learned by DLOCCNet that can output a state that achieves a fidelity higher than that of dynamic DEJMPS and close to the highest possible fidelity. The DLOCCNet for distillation involves Alice and Bob performing local operations on their own qubits independently, using two copies of some S states in the initial round. After their local measurements, they communicate their outcomes through classical channel. The measured qubits will be reset when both Alice and Bob get $0$ from computational basis measurements, followed by adding an additional one-copy of some S states in these qubits. Building on this framework, we present a simplified circuit that achieves performances comparable to universal DLOCCNet, which means each local operation is the universal unitary. Furthermore, we analytically derive the fidelity expression for the distilled state.
Details of our protocol are given in Fig.~\ref{fig:s-n} $(b)$. The final fidelity achieved by this protocol is compared with that achieved by the dynamic DEJMPS protocol in Fig.~\ref{fig:s-n} $(c)$ and $(d)$. Fig.~\ref{fig:s-n} $(d)$ shows that the protocol learned by DLOCCNet performs better than dynamic DEJMPS when using $5$-copy of some S states. Considering the general case, i.e., $n$-copy of some S states, Fig.~\ref{fig:s-n} $(c)$ implies that the protocol learned by dynamic LOCCNet achieves a higher fidelity than that of dynamic DEJMPS for different noisy parameters of S state. In this $n$-copy case, the dynamic DEJMPS's fidelity to the ebit is given by 
$$
f_n^{\rm ddejmps}=
\begin{cases}
    \frac{f_1f_{n-1}}{1-f_1+(2f_1-1)f_{n-1}}, &n=3k-1,\\
    \frac{f_1f_{n-1}}{f_{n-1}+f_1-1},&n=3k,\\
    \frac{f_{n-1}}{2f_{n-1}-1},&n=3k+1,
\end{cases}
$$
for $k=1,2,\cdots$, where $f_1=\frac{1+\gamma}{2}$. While 
\begin{equation}
    f_n^{\rm dloccnet}=\frac{f_1f_{n-1}}{1-f_1+(2f_1-1)f_{n-1}}, n\geq 3,
\end{equation}
is fidelity achieved by the protocol learned by dynamic LOCCNet using $n$-copy of S state, where $f_1=\frac{1+\gamma}{2}$, and 
\begin{equation}
    f_2^{\rm dloccnet}=\frac{1}{2}\left(1+\sqrt{-\gamma(\gamma-2)}\right).
\end{equation}
\ud{The above formulas offer a systematic method for computing the fidelity of the distilled states after each round, and the detailed derivation can be found in the Supplementary Material~\ref{app:s}}.

\begin{figure*}[t]
    \centering
    \includegraphics[width=\linewidth]{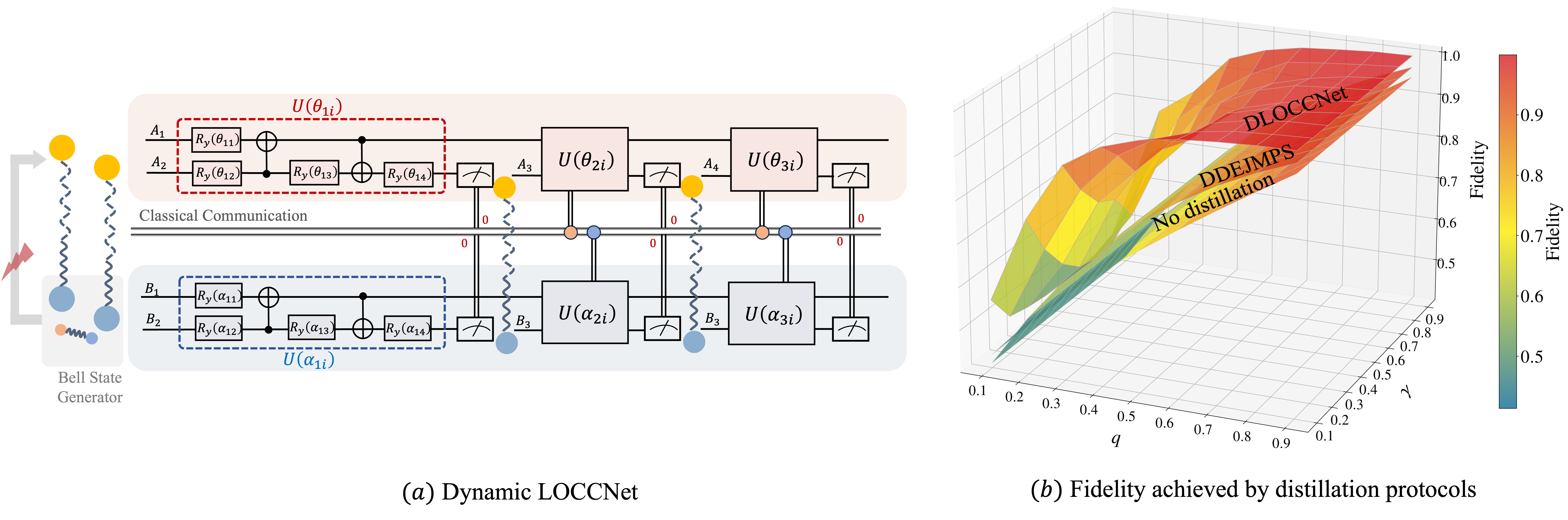}
    \caption{\textbf{Fidelity achieved by distillation methods for maximally entangled states affected by unilocal GAD channel.} $(a)$ is the simplified circuit learned by dynamic LOCCNet for $4$ copies of noisy states, where $U(\bm \theta)$ and $ U(\bm \alpha)$ are composed by $R_y$ rotation and CNOT gates. $(b)$ shows the fidelity achieved using two types of distillation protocols.}
    \label{fig:gad}
\end{figure*}
Beyond the erasure channel, we investigate distillation performance under amplitude damping (AD) and generalized amplitude damping (GAD) channels~\cite{fletcher2008channel,fujiwara2004estimation}, which model energy dissipation processes such as spontaneous emission. The AD channel transforms a state $\rho$ according to
\begin{equation} 
    \cN_{\rm ad}(\rho) = E_0\,\rho\,E_0^\dagger + E_1\,\rho\,E_1^\dagger,\label{eq:ad}
\end{equation}
where the Kraus operators are defined as $E_0 = \ketbra{0}{0} + \sqrt{1-\gamma}\ketbra{1}{1}$, $E_1 = \sqrt{\gamma}\,|0\rangle\langle 1|$, with $\gamma\in [0,1]$ representing the dissipation rate. The GAD channel provides a more comprehensive model, incorporating thermal effects through:
\begin{equation}
    \cN_{\rm gad}(\rho)=\sum_{i=0}^3 E_i^{\rm g}\rho (E_i^{\rm g})^{\dagger},
\end{equation}
with Kraus operators
\begin{align}
   & E_0^{\rm g}=\sqrt{q}\begin{pmatrix}
    1&0\\
    0& \sqrt{\gamma}
    \end{pmatrix},&
    E_1^{\rm g}=\sqrt{q}\begin{pmatrix}
    0&\sqrt{1-\gamma}\\
    0& 0
    \end{pmatrix},\\\nonumber
    \vspace{0.7pt}
    &E_2^{\rm g}=\sqrt{1-q}\begin{pmatrix}
    \sqrt{\gamma}&0\\
    0& 1
    \end{pmatrix},&
    E_3^{\rm g}=\sqrt{1-q}\begin{pmatrix}
    0&0\\
    \sqrt{1-\gamma}& 0
    \end{pmatrix},
\end{align}
where $q\in[0,1]$ represents the temperature of the environment.

We conduct two complementary studies to evaluate the DLOCCNet. First, we compare dynamic DEJMPS and DLOCCNet using 2-to-1 distillation for 4 copies of Bell states under unilocal GAD channel. In this context, a unilocal GAD channel means that after Bob generates the Bell state with the generator, the half of state is transmitted to Alice through a single-qubit GAD channel, as shown in Fig.~\ref{fig:gad}. Second, we examine the $n$-copy scenario under bilocal AD channel to assess scalability.

The Bell state $|\Phi^{+}\rangle$ shared between Alice and Bob undergoes a GAD channel action only on Alice's part, obtaining the mix state.
\begin{equation}
\sigma_{\rm gad}^{\rm uni}=\cN_{\rm gad}\otimes\cI\bigl(|\Phi^{+}\rangle\langle\Phi^{+}|\bigr).
\end{equation}
In this context, the DLOCCNet for distillation proceeds that Alice and Bob first apply local operations $U(\bm \theta_{1\bm i}),U(\bm \alpha_{1\bm i})$ to their respective two-qubit subsystems. After measuring their second qubits and communicating the results classically, they perform conditional reset operations when both obtain $0$ measurement outcomes. Subsequently, they introduce an additional GAD-corrupted Bell state and execute the next round of local operations $U(\bm \theta_{2\bm i}),U(\bm \alpha_{2\bm i})$. Our DLOCCNet implementation employs simplified local unitaries, $U(\bm \theta_{\bm j \bm i}),U(\bm \alpha_{\bm j\bm i})$, consisting of single-qubit $R_y(\bm \theta_{\bm j\bm i}),R_y(\bm \alpha_{\bm j\bm i})$ rotations and two-qubit CNOT gates, as illustrated in Fig.~\ref{fig:gad} $(a)$. The resulting fidelity comparison with dynamic DEJMPS is presented in Fig.~\ref{fig:gad} $(b)$, demonstrating superior performance of our method for $4$-copy distillation.

For the $n$-copy case under the bilocal AD channel, each qubit of the Bell state $|\Phi^{+}\rangle$ experiences AD channel $\cN_{\rm ad}$ is
\begin{equation}
\sigma_{\rm ad}^{bi}=\cN_{\rm ad}\otimes\cN_{\rm ad}\bigl(|\Phi^{+}\rangle\langle\Phi^{+}|\bigr).
\end{equation}

Fig.~\ref{fig:exp-ad} presents the fidelity results, clearly demonstrating that DLOCCNet consistently outperforms dynamic DEJMPS across various damping parameters $\gamma$, achieving superior distillation fidelity in all tested configurations.

\begin{figure}[h]
    \centering
    \includegraphics[width=0.8\linewidth]{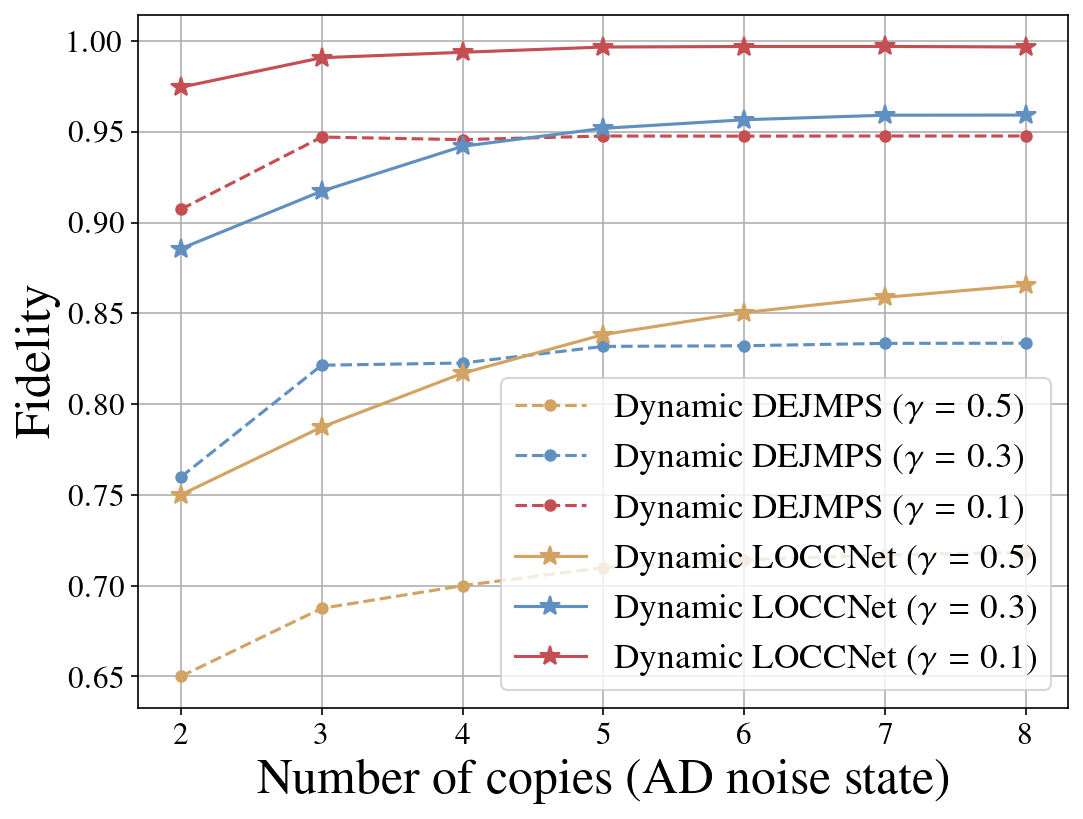}
    \caption{\textbf{Fidelity achieved by distillation methods for various of copies of maximally entangled states affected by bilocal AD channel.}}
    \label{fig:exp-ad}
\end{figure}

\begin{figure*}[t]
    \centering
    \includegraphics[width=0.8\linewidth]{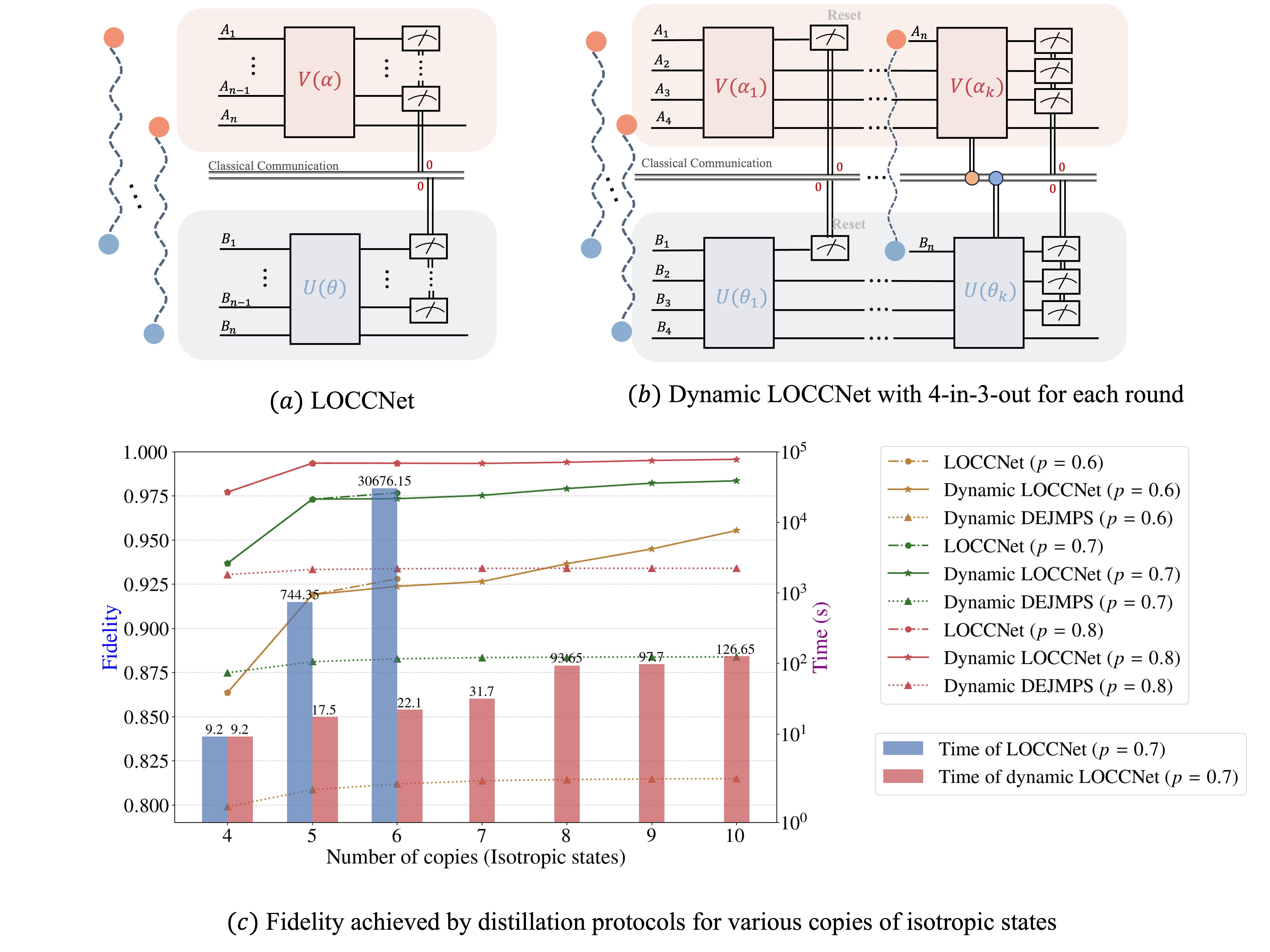}
    \caption{\textbf{Fidelity achieved by distillation methods for maximally entangled states affected by bilocal depolarizing channel.} $(a)$ is the distillation protocol learned by LOCCNet. $(b)$ shows the isotropic state distillation via dynamic LOCCNet method using $n$ copies. $(c)$ demonstrates the fidelity achieved by distillation protocols for various copies of isotropic states.}
    \label{fig:iso-n}
\end{figure*}

The depolarizing quantum channel~\cite{king2003capacity} is another fundamental type of quantum noise that describes a process in which a quantum state loses its coherence and becomes maximally mixed with a certain probability. For a $d$-dimensional quantum system, the depolarizing channel $\cN_{\rm dep}$ acting on a state $\rho$ is given by
\begin{align}\label{eq:dep}
    \cN_{\rm dep}(\rho)=p\rho+(1-p)\frac{I}{d},
\end{align}
where $\frac{I}{d}$ is the maximally mixed state in dimension $d$. In the special case of a two-qubit system, the isotropic states can be viewed as the result of applying the bilocal depolarizing channel to the ebit $\ket{\Phi^+}$. The state is given by
\begin{align}
\rho_{\rm iso}=p\ketbra{\Phi^+}{\Phi^+}+(1-p)\frac{I}{4},
\end{align}
with $p\in[0,1]$, which illustrates how depolarizing noise affects entanglement and quantum correlations.

The distillation protocols obtained by LOCCNet for four copies of some isotropic state have been well studied in~\cite{zhao2021practical}, achieving empirically optimal fidelity for the four-copy case. Although using additional copies of isotropic states can yield a state with enhanced fidelity, designing the corresponding protocol with LOCCNet requires exponential time. To address this challenge, we extend the distillation to the $n$-copy regime using DLOCCNet. This approach not only achieves fidelity comparable to that of the LOCCNet in significantly less time (as shown in Fig.~\ref{fig:iso-n}), but is also suitable for the small-scale hardware.
Additionally, we also explore the $16\rightarrow 1$ distillation using the iterative method. This method is widely recognized as a highly effective technique among distillation strategies, demonstrating superior performance compared to alternative methods when evaluated with the same number of copies.
Remarkably, DLOCCNet demonstrates similar performance to the iterative method (as shown in Fig.~\ref{fig:iso-16}), highlighting its effectiveness in practical applications.

\begin{figure*}[t]
    \centering
    \includegraphics[width=\linewidth]{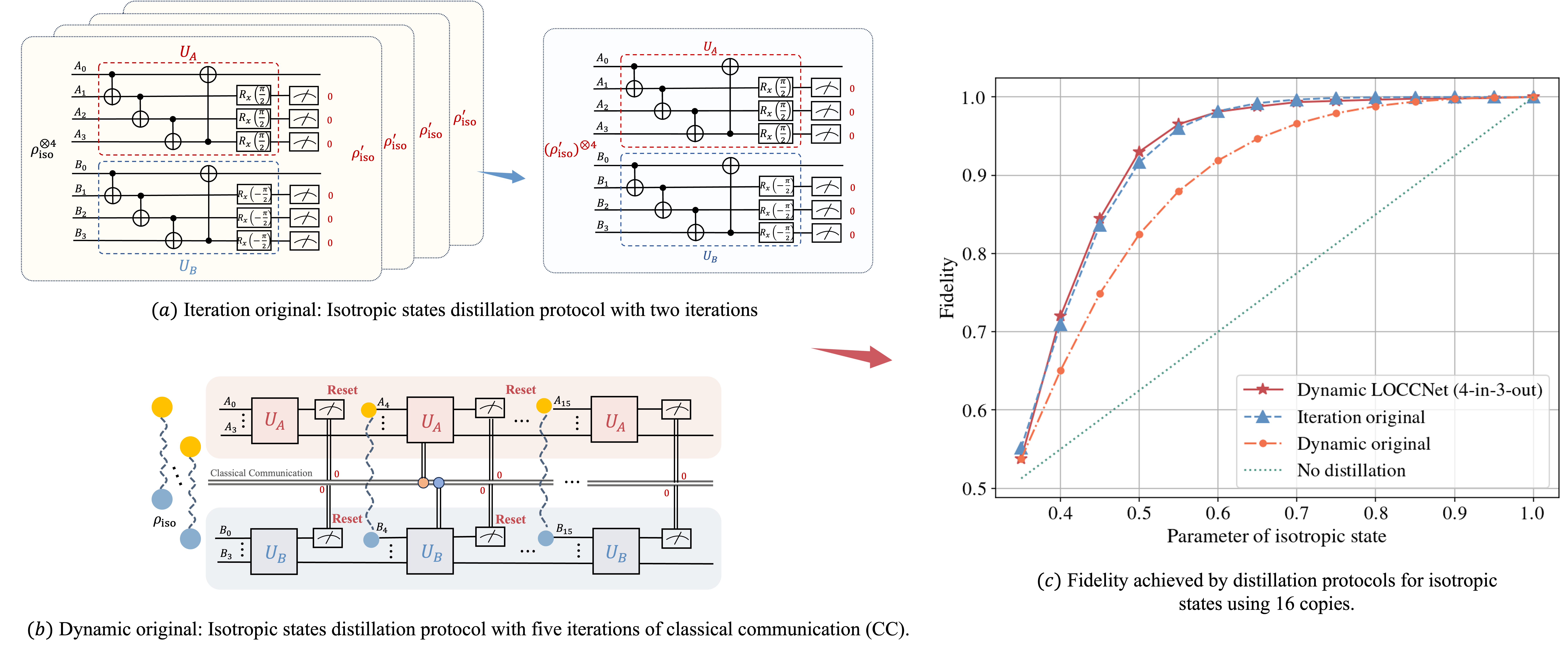}
    \caption{\textbf{Isotropic states distillation using different kinds of method.} $(a)$ Iteration original protocol shows two iteration of original LOCCNet protocol with fixed parameters. $(b)$ Dynamic original protocol uses dynamic method with five iterations of classical communication (CC), where $U_A, U_B$ are circuits shown in $(a)$. $(c)$ Fidelity achieved by distillation protocols for isotropic states using $16$ copies. The blue dashed line shows the fidelity between isotropic state with Bell state. The red solid line depicts the performance of the dynamic LOCCNet shown in Fig.~\ref{fig:iso-n} $(b)$, and the blue dash line is the result of iterative method in $(a)$. The orange line shows the fidelity achieved corresponding to the protocol in $(b)$.}
    \label{fig:iso-16}
\end{figure*}

The final fidelity achieved by DLOCCNet, using various copies of the isotropic state, is compared to that of LOCCNet in Fig.~\ref{fig:iso-n}. The optimal fidelity attained by LOCCNet is around 0.975 by consuming 6 copies of isotropic states with $p=0.7$. Training such a $6\rightarrow1$ distillation protocol takes $3\times 10^4$ seconds on a personal computer. Due to the exponential increase in training time, it is basically impossible to train a $7\rightarrow1$ entanglement distillation protocol.
In contrast, DLOCCNet easily surpasses this limitation, achieving fidelity exceeding 0.98 by consuming additional copies of noisy states. Our implementation of DLOCCNet employs a measurement scheme with a 4-in-3-out configuration at each round: for four copies of the input noisy state, only one qubit per system is measured in each round. If both measurement results are zero, the remaining qubits are passed on to the subsequent round. In the final round, three qubits per system are measured, and id all measurement outcomes are zero, the remaining distilled state is obtained (as shown in Fig.~\ref{fig:iso-n} $(b)$). With this ansatz, the training time for the $10\rightarrow1$ protocol only takes 126 seconds, which is a tremendous improvement compared to LOCCNet. \ud{Note that DLOCCNet solves large-scale problems by recursively training smaller PQCs. Consequently, its training complexity increases polynomially, rather than exponentially as constrained by previous methods.}

Considering the same number of copies, we explore the $16\rightarrow 1$ distillation process. In addition to the comparison with dynamic technique, we also evaluate the performance of the iterative method (as shown in Fig.~\ref{fig:iso-16} $(a)$), which is regarded as the state-of-the-art in distillation~\cite{chen2024optimum}. For a given fixed distillation protocol (here we use the the well-performing protocol proposed by Zhao et al.~\cite{zhao2021practical}), the iterative method can achieve higher fidelity compared to the dynamic method, because the cumulative effect of the fidelity improvement over consecutive rounds of the dynamic method increases the fidelity gap between the input states in each round. Compared with DLOCCNet shown in Fig.~\ref{fig:iso-n} $(b)$, the fidelity achieved by three different distillation methods for isotropic states using 16 copies is illustrated in Fig.~\ref{fig:iso-16} $(c)$. The results indicate that the distillation performance of the DLOCCNet with 4-in-3 out, is comparable to that of the iterative method depicted in Fig.~\ref{fig:iso-16} $(a)$. The fidelity of the iterative method to the ebit is given by
\begin{equation}
     f_i^{\rm itr}=\frac{1-4f_{i-1}+6f_{i-1}^2}{3-8f_{i-1}+8f_{i-1}^2},\; i\geq 1,
\end{equation}
where $f_{i-1}$ is the fidelity after $i-1$ iterations and $f_0=\frac{1+3p}{4}$. For the case of dynamic original in Fig.~\ref{fig:iso-16} $(b)$, the fidelity of the output state after $i$ iterations, denoted as $f_{n}^{\rm dyn}$, is given by the recurrence relation
\begin{align}
  f_i^{\rm dyn}=\frac{1-2f_0+f_0^2-3f_0f_{i-1}+12f_0^2f_{i-1}}{3-3f_0+f_{i-1}-8f_0f_{i-1}+16f_0^2f_{i-1}},\; i\geq 2,
\end{align}
where $f_{i-1}$ is the fidelity after $i-1$ iterations and $f_1=\frac{1-2p+9p^2}{4-8p+12p^2}$, $f_0=\frac{1+3p}{4}$. 
The derivation of the above two fidelities can be seen in the Supplementary Material~\ref{app:iso}.

The formulas presented above establish that the iterative distillation of isotropic states, using either the specified protocol for the dynamic or iterative methods, provides a systematic approach for calculating the fidelity of the distilled states after each round. It is important to note that the number of copies of the initial isotropic state is $4^i$, with $N_{i}^{\rm ori}=\frac{4^i-1}{3}$ times protocol using Fig.~\ref{fig:iso-16} $(a)$. In contrast, referring to Fig.~\ref{fig:iso-16} $(b)$, the number of copies of the initial isotropic state is $3i+1$, resulting in $N_{i}^{\rm dyn}=i$ iterations of the dynamic protocol. Notably, the iterative method requires the input copies to be the exponentiation of four, whereas the dynamic architecture removes this constraint, supporting arbitrary input sizes. This flexibility allows resource-efficient distillation that can be tailored to experimental or computational needs.

Furthermore, we also extend our analysis to the distillation of maximally entangled states in higher-dimensional system ($d>2$) that are affected by depolarizing channel. Detailed experimental result is provided in the Supplementary Material~\ref{app:other-distill-discri}.

\subsection{Distributed state discrimination}
Quantum state discrimination (QSD)~\cite{bae2015quantum, barnett2009quantum}, distinguishing one physical configuration from another, is a fundamental task in quantum information theory, which underpins the security of quantum cryptography~\cite{gisin2002quantum, bennett2014quantum}, the reliability of quantum communication~\cite{gisin2007quantum} and hypothesis testing~\cite{hayashi2006study}. QSD captures the boundary between what is knowable and unknowable about a quantum system, making it central to applications, including quantum data hiding~\cite{divincenzo2002quantum} and dimension witness~\cite{gallego2010device}.

Distributed state discrimination has been well studied in prior works~\cite{zhu2025entanglement,walgate2000local, hayashi2006bounds, duan2007distinguishing, li2017indistinguishability}. In particular, Walgate et al.~\cite{walgate2000local} demonstrated that when multiple parties share a single copy of a quantum state prepared in one of two orthogonal pure states, perfect discrimination can be achieved using LOCC operations. However, how to design a practical, implementable quantum circuit for distinguishing states on a quantum device remains a significant challenge.

Here, we apply our DLOCCNet framework to distinguish quantum states and investigate how the probability of success improves with respect to the increase in copies of states. We utilize multiple copies of the states without increasing the circuit width, as shown in Fig.~\ref{fig:discrimination}. Maintaining the same bit count while improving performance offers a more accessible route for near-term quantum devices and broadens the potential for practical applications.

\begin{figure}[t]
    \centering
    \includegraphics[width=\linewidth]{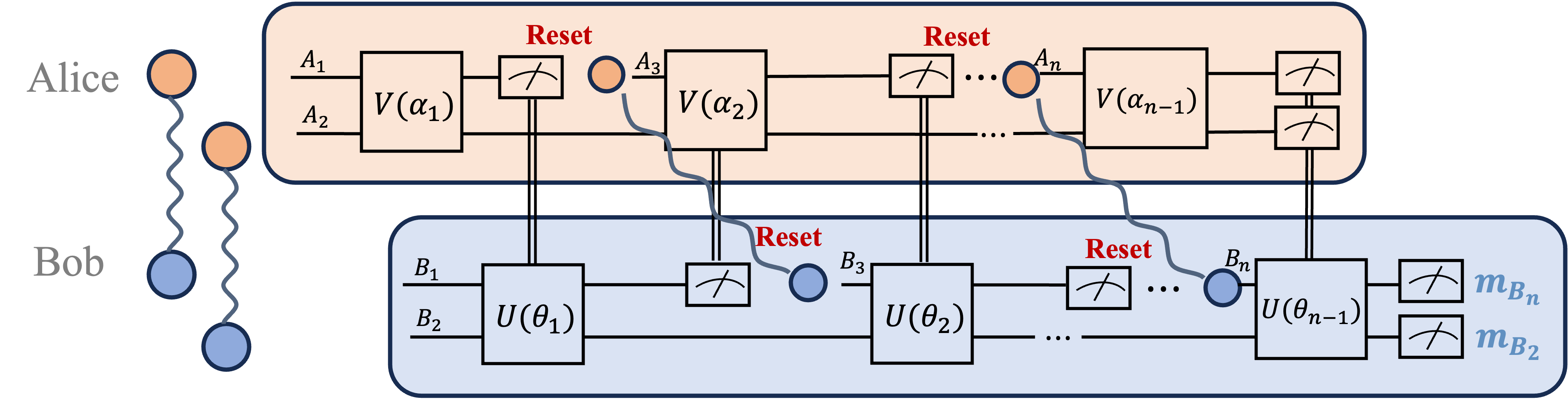}
    \caption{\textbf{Illustration of discrimination using dynamic LOCCNet.} With preshared states, Alice sends her measurement outcome to Bob in the first round, which determines Bob's local operation. Then Bob sends his measurement outcome to Alice. In the final round, after using the $n$-th copy of the preshared state, Bob infers the state based on the measurement outcome $m_{B_2}$ from his second register.}
    \label{fig:discrimination}
\end{figure}

Specifically, we focus on distinguishing between two states: the Bell state $\Phi_0=\proj{\Phi^+}$ and a noisy maximally entangled state $\Phi_1=\cN\ox\cN(\proj{\Phi^-})$ with $\ket{\Phi^-}=\frac{1}{\sqrt{2}}(\ket{00}-\ket{11})$. Alice and Bob share $n$ copies of $\Phi_0$ or $\Phi_1$, the goal is to distinguish whether the shared state is $\Phi_0$ or $\Phi_1$ through LOCC.

To develop an LOCC protocol for distinguishing a given state between $\Phi_0$ and $\Phi_1$, we utilize the structure shown in Fig.~\ref{fig:discrimination}, which consumes $n$ copies of states. With preshared states, Alice first applies local PQC on her part and performs measurement on the first register. Denote the measurement result as $m_{A_1}\in\{0,1\}$, which is then communicated to Bob via a classical channel. Upon receiving $m_{A_1}$, Bob applies a unitary conditioned on Alice's measurement to his systems. Subsequently, Bob performs a measurement on his first register, obtaining outcome $m_{B_1}\in\{0,1\}$, which he sends back to Alice. Both parties then reset the first register and repeat the previous processes, including applying the circuit, making measurements, communicating, and resetting systems. When all $n$ copies are processed, Alice measures all her systems and transmits the results $m_{A_n}, m_{A_2}$ to Bob. Bob, after applying the corresponding unitaries based on Alice's classical information, measures all his registers as well and records the outcomes $m_{B_n}, m_{B_2}$. We define that when Bob's measurement outcome on the second register, $m_{B_2}$, equals $0$ or $1$, the preshared state is inferred to be $\Phi_0$ or $\Phi_1$, respectively. Since our aim is to maximize the success probability, or equivalently, minimize the failure probability, we can define the cost function as 
\begin{equation}
    L= P(1 \mid \Phi_0) + P(0 \mid \Phi_1).
\end{equation}
where $P(j\mid \Phi_k)$ is the probability that the final measurement yields outcome $j$ given input state $\Phi_k$.

Fig.~\ref{fig:discri} presents a comparative analysis of discrimination performance for $\Phi_0=\proj{\Phi^+}$ and $\Phi_1=\cN\ox\cN(\proj{\Phi^-})$ across two noise scenarios using varying copy numbers. The left and right subplot in Fig.~\ref{fig:discri} corresponds to $\cN$ to be the amplitude damping noise and dephasing noise, respectively. The amplitude damping noise has been defined in Eq.~\eqref{eq:ad}. The dephasing channel $\cN_{\rm de}$ refers to the process that destroys the phase coherence of a quantum system, which maps the state $\rho$ to 
\begin{equation}
    \cN_{\rm de}(\rho) = (1-p) \rho + p\sigma_Z\rho\sigma_Z,
\end{equation}
where $p$, $\sigma_Z$ refers to the noise level and Pauli-Z matrix, respectively. Besides, we also consider the noise to be a depolarizing channel $\cN_{\rm dep}$, which is defined in Eq.~\eqref{eq:dep} and the numerical results are shown in Supplementary Materials~\ref{app:other-distill-discri}.

The results demonstrate a substantial improvement in success probability as the number of state copies increases. Crucially, the comparison between different copy numbers implementations reveals that DLOCCNet can enhance discrimination success probability without requiring additional circuit width, effectively leveraging quantum parallelism within the existing system constraints. This finding highlights a key advantage of our approach: the ability to exploit multiple quantum copies for improved performance while maintaining practical implementability through constant circuit complexity.

\begin{figure}[t]
    \centering
    \includegraphics[width=\linewidth]{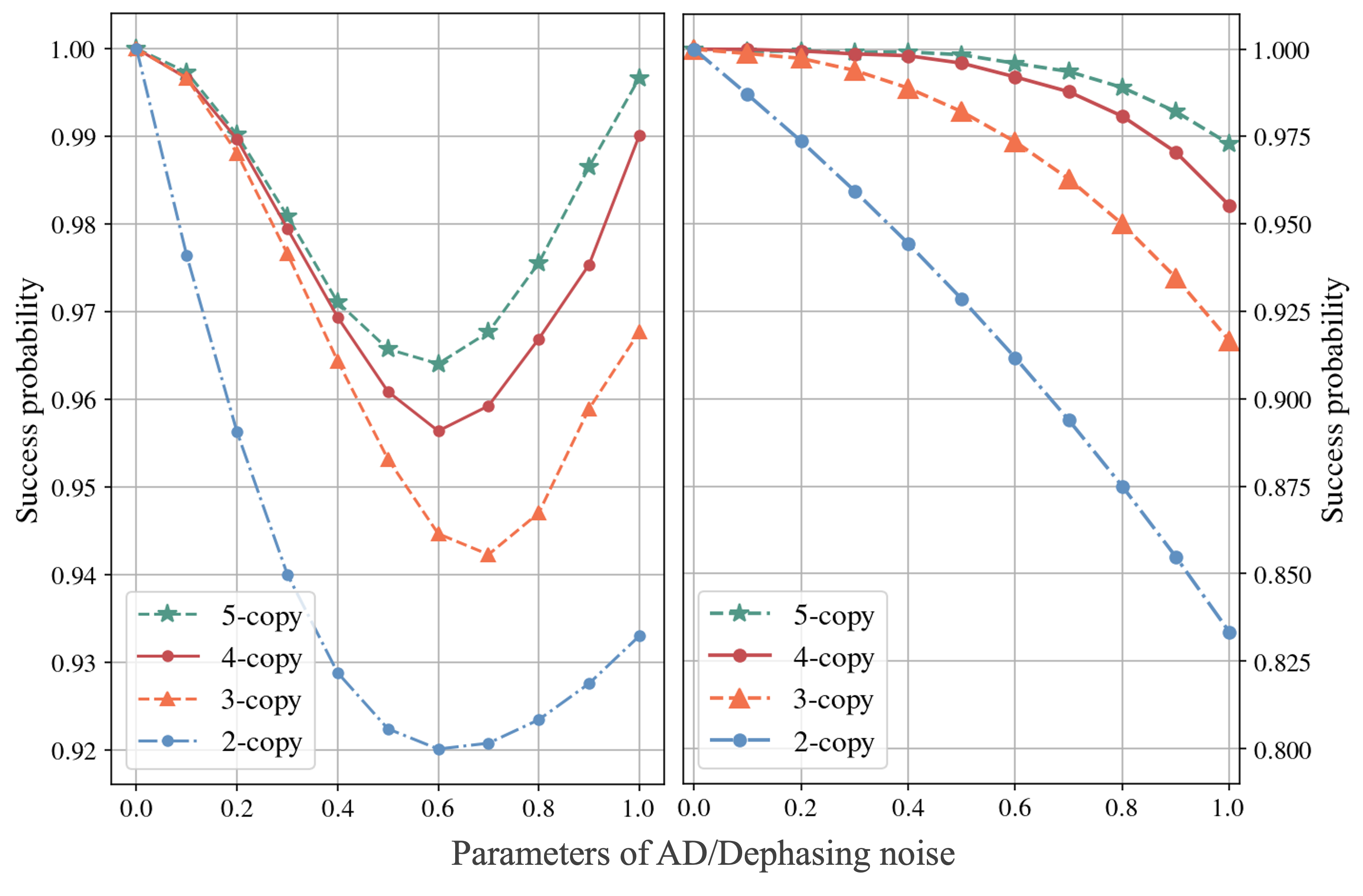}
    \caption{\textbf{Average success probability of distinguishing a Bell state and a noisy Bell state using different copy numbers.} }
    \label{fig:discri}
\end{figure}

\section{Discussion}
In this work, we introduce a new framework DLOCCNet, which is a general and flexible framework for simulating and designing LOCC protocols, and demonstrate its applications in entanglement distillation and distributed state discrimination. For entanglement distillation, our framework enables scalable protocols that demonstrate significant improvements across various noise models: it outperforms existing methods for Bell states affected by erasure channels, achieves fidelity comparable to LOCCNet with substantially faster training for depolarizing channels, and attains higher fidelity than previous protocols in the case of amplitude damping channels. Crucially, DLOCCNet could design an LOCC protocol that involves a large number of systems, which is not possible with the previous framework. In distributed state discrimination, we study optimized circuits that maximize success probabilities across diverse scenarios. The result demonstrates a substantial improvement in success probability as the number of state copies increases. Compared to previous work, DLOCCNet offers superior scalability and faster training, advancing both the practical utility and theoretical understanding of LOCC protocols. 

The advantage of our framework can be understood as the trade-off between expressibility and trainability. While larger PQCs can undoubtedly express more operations by covering additional systems, they become exponentially more difficult to train. This training difficulty arises primarily from the barren plateau~\cite{mcclean2018barren}, which refers to the phenomenon where the gradient drops exponentially with respect to the size of the quantum systems. In our framework, we solve large problems by decomposing it into recursively trainable small PQCs. Although this approach limits expressibility, the training complexity increases polynomially rather than exponentially, making it possible to solve large-sized problems. In addition, since the PQCs used in DLOCCNet are small, our method naturally mitigates the detrimental effects of barren plateaus. \ud{In this work, the processes of PQCs and cost estimation are all simulated on a classical computer. Here, we emphasize that such processes can also be implemented on a quantum device.}

Due to the practicality of the designed LOCC protocol, one interesting further research is to implement the designed LOCC protocols on quantum devices to check the performance. Note that the framework only focuses on the bipartite case, another interesting research topic in the future is to design the multipartite LOCC protocols.

\textbf{Acknowledgements}--- 
\ud{We would like to thank Ranyiliu Chen and Mingrui Jing for their helpful suggestions.} This work was partially supported by the National Key R\&D Program of China (Grant No.~2024YFB4504004), the National Natural Science Foundation of China (Grant. No.~12447107), the Guangdong Provincial Quantum Science Strategic Initiative (Grant No.~GDZX2403008, GDZX2403001), the Guangdong Provincial Key Lab of Integrated Communication, Sensing and Computation for Ubiquitous Internet of Things (Grant No.~2023B1212010007), the Quantum Science Center of Guangdong-Hong Kong-Macao Greater Bay Area, and the Education Bureau of Guangzhou Municipality.
\bibliography{ref.bib}

\clearpage

\vspace{2cm}
\onecolumngrid
\vspace{2cm}
\begin{center}
{\textbf{\large Supplementary Materials of \\
\vspace{10pt}
\centering Dynamic LOCC Circuits for Automated Entanglement Manipulation}}
\end{center}
\renewcommand{\theequation}{S\arabic{equation}}
\renewcommand{\theproposition}{S\arabic{proposition}}
\renewcommand{\thedefinition}{S\arabic{definition}}
\renewcommand{\thefigure}{S\arabic{figure}}
\setcounter{equation}{0}
\setcounter{table}{0}
\setcounter{section}{0}
\setcounter{figure}{0}

\appendix


\section{Analysis of S states distillation}\label{app:s}
\subsection{Dynamic DEJMPS distillation protocol for S state}

\begin{lemma}\label{lem:de-s}
Let $\rho_i\in\{\rho_{\rm S}^{f_i}\}_i$ be an S state defined as
\begin{align}
  \rho_i=(2f_i-1)\Phi^++2(1-f_i)\ketbra{00}{00},\label{eq:s}
\end{align}
where $\Phi^+$ is the maximally entangled state and $f_i=\frac{1+\gamma_i}{2}$ with $\gamma_i$ being the parameter characterizing the S state. Consider the input state given by 
\begin{align}
  \sigma_{\rm in}=\rho_j\otimes\rho_{k}
\end{align}
where $\rho_{j},\rho_k\in\{\rho_{\rm S}^{\gamma_i}\}_i$. Then the fidelity of distilled state by DEJMPS with the maximally entangled state $\Phi^+$ is 
\begin{align}
    f=\frac{f_jf_k}{1-f_j-f_k+2f_jf_k},
\end{align}
where $f_j=\frac{1+\gamma_j}{2}, f_k=\frac{1+\gamma_k}{2}$ with
$\gamma_j, \gamma_k$ being the parameters characterizing $\rho_j$ and $\rho_k$, respectively.  
\end{lemma}
\begin{proof}
 Given an input state $\sigma_{\rm in}=\rho_j\otimes \rho_k$ with $\rho_j=(2f_j-1)\Phi^++2(1-f_j)\ketbra{00}{00}$, and $\rho_k=(2f_k-1)\Phi^++2(1-f_k)\ketbra{00}{00}$, it is easy to obtain the distilled state after DEJMPS protocol is 
 \begin{align}
     \sigma_{\rm out}=a\Phi^++(1-a)\Psi^{+}+\frac{a-1}{2}[(\ket{00}+\ket{11})(\bra{01}+\bra{10})+(\ket{01}+\ket{10})(\bra{00}+\bra{11})],\label{eq:output-s-1-de}
 \end{align}
 where 
 \begin{align}
     \Phi^+&=\frac{1}{2}(\ketbra{00}{00}+\ketbra{00}{11}+\ketbra{11}{00}+\ketbra{11}{11}),\\
     \Psi^+&=\frac{1}{2}(\ketbra{01}{01}+\ketbra{01}{10}+\ketbra{10}{01}+\ketbra{10}{10}),\\
     a&=\frac{f_jf_k}{1-f_j-f_k+2f_jf_k},\\
     f_j&=\frac{1+\gamma_j}{2},\;f_k=\frac{1+\gamma_k}{2}
 \end{align}
 with $\gamma_j, \gamma_k$ being the parameters characterizing $\rho_j$ and $\rho_k$, respectively. Therefore, the fidelity of the distilled state by DEJMPS with $\Phi^+$ is 
 \begin{align}
     f=a=\frac{f_jf_k}{1-f_j-f_k+2f_jf_k}.
 \end{align}
\end{proof}

\begin{proposition}\label{prop:s-de}
 Let $\rho\in\{\rho_{\rm S}^{\gamma_i}\}_i$ be an S state defined as
  \begin{align}
      \rho=(2f_1-1)\Phi^++2(1-f_1)\ketbra{00}{00},\label{eq:s}
  \end{align}
  where $\Phi^+$ is the maximally entangled state and $f_1=\frac{1+\gamma}{2}$ with $\gamma$ being the parameter characterizing the S state. Consider the dynamic DEJMPS protocol as depicted in Fig.~\ref{fig:s-n} $(a)$, the fidelity of the output state is 
$$
f_n^{\rm ddejmps}=
\begin{cases}
    \frac{f_1f_{n-1}}{1-f_1+(2f_1-1)f_{n-1}}, &n=3k-1,\\
    \frac{f_1f_{n-1}}{f_{n-1}+f_1-1},&n=3k,\\
    \frac{f_{n-1}}{2f_{n-1}-1},&n=3k+1,\label{eq:s-dyn-de}
\end{cases}
$$
for $k=1,2,\cdots$, where $f_1=\frac{1+\gamma}{2}$, and $n$ is the copy number of $S$ state.     
\end{proposition}
\begin{proof}
    Before proving the fidelity being Eq.~\eqref{eq:s-dyn-de}, we first show that the output state after three-round dynamic DEJMPS still belongs to the set of S states. 

For any $\sigma^{(1)}_{\rm in}=\rho_j\otimes\rho_k,\;\forall\rho_j,\rho_k\in\{\rho_{\rm S}^{f_i}\}_i$, the output state after DEJMPS is $\sigma^{(1)}_{\rm out}=a\Phi^++(1-a)\Psi^{+}+\frac{a-1}{2}[(\ket{00}+\ket{11})(\bra{01}+\bra{10})+(\ket{01}+\ket{10})(\bra{00}+\bra{11})]$ in Eq.~\eqref{eq:output-s-1-de} based on the proof of Lemma~\ref{lem:de-s}. Denote
\begin{align}
    S_1:=\{\eta_i|\eta_i=a_i\Phi^++(1-a_i)\Psi^{+}+\frac{a_i-1}{2}[(\ket{00}+\ket{11})(\bra{01}+\bra{10})+(\ket{01}+\ket{10})(\bra{00}+\bra{11})]\},
\end{align}
then for any $\sigma^{(2)}_{\rm in}=\eta_j\otimes\rho_j,\;\forall\eta_j\in S_1,\rho_j\in\{\rho_{\rm S}^{f_i}\}_i$, we will obtain the output state 
\begin{align}
    \sigma^{(2)}_{\rm out}=b\Phi^{+}+(1-b)\Psi^-+\frac{(1-b)\mathrm{i}}{2}[(\ket{00}+\ket{11})(\bra{01}-\bra{10})+(\ket{10}-\ket{01})(\bra{00}+\bra{11})]
\end{align}
after DEJMPS protocol, where 
 \begin{align}
     \Phi^+&=\frac{1}{2}(\ketbra{00}{00}+\ketbra{00}{11}+\ketbra{11}{00}+\ketbra{11}{11}),\\
     \Psi^-&=\frac{1}{2}(\ketbra{01}{01}-\ketbra{01}{10}-\ketbra{10}{01}+\ketbra{10}{10}),\\
     b&=\frac{a_jf_j}{-1+a_j+f_j},\\
     f_j&=\frac{1+\gamma_j}{2}
 \end{align}
 with $\gamma_j$ being the parameters characterizing $\rho_j$. Therefore, the fidelity of the distilled state by the secound round of DEJMPS with $\Phi^+$ is 
 \begin{align}
     f^{(2)}=b=\frac{a_jf_j}{-1+a_j+f_j}.
 \end{align}
 Similarly, denote
 \begin{align}
    S_2:=\{\zeta_j|\zeta_j=b_j\Phi^++(1-b_j)\Psi^-+\frac{(1-b_j)\mathrm{i}}{2}[(\ket{00}+\ket{11})(\bra{01}-\bra{10})+(\ket{10}-\ket{01})(\bra{00}+\bra{11})]\},
\end{align}
then after the third round of DEJMPS protocol, the output state will be
\begin{align}
    \sigma^{(3)}_{\rm out}=c\Phi^++(1-c)\ketbra{00}{00}
\end{align}
for any input state $\sigma^{(3)}_{\rm in}=\zeta_m\otimes\rho_m,\;\forall\zeta_m\in S_2,\rho_j\in\{\rho_{\rm S}^{f_i}\}_i$, where $c=\frac{1}{2b_m-1}$, and the fidelity after the third round is
\begin{align}
    f^{(3)}=\frac{b_m}{2b_m-1}.
\end{align}
Note that the output state, $\sigma^{(3)}_{\rm out}$, of the third round belongs to the S state family, which implies that the dynamic DEJMPS process repeats every three iterations as a cycle.

Then, for $n=3k-1$, there is $f_n^{\rm ddejmps}=\frac{f_1f_{n-1}}{1-f_1+(2f_1-1)f_{n-1}}$ according to Lemma~\ref{lem:de-s}. When $n=3k$, we have 
\begin{align}
    f_n^{\rm ddejmps}&=\frac{a_{n-1}f_1}{-1+a_{n-1}+f_1}\\
    &=\frac{f_1f_{n-1}}{f_{n-1}+f_1-1}
\end{align}
because of $a_{n-1}=f_{n-1}$. And
\begin{align}
    f_n^{\rm ddejmps}&=\frac{b_{n-1}}{2b_{n-1}-1}\\
    &=\frac{f_{n-1}}{2f_{n-1}-1}
\end{align}

\end{proof}

\subsection{Dynamic LOCCNet distillation protocol for S state}
\begin{lemma}\label{lem:rho-s}
  Let $\rho_i\in\{\rho_{\rm S}^{\gamma_i}\}_i$ be an S state defined as
  \begin{align}
      \rho_i=(2f_i-1)\Phi^++2(1-f_i)\ketbra{00}{00},\label{eq:s}
  \end{align}
  where $\Phi^+$ is the maximally entangled state and $f_i=\frac{1+\gamma_i}{2}$ with $\gamma_i$ being the parameter characterizing the S state. Consider the input state given by 
  \begin{align}
      \eta_{\rm in}=\sigma_j\otimes\rho_{\star}
  \end{align}
  where $\rho_{\star}\in\{\rho_{\rm S}^{\gamma_i}\}_i$ and 
  \begin{align}
      \sigma_j\in\{\sigma_k\}_k:=\left\{\sigma_k=\Phi^++a_k\left(\ketbra{00}{11}+\ketbra{11}{00}\right)\Big|a_k\in[-1,0]
      \right\}.
  \end{align}
Then the output state $\eta_{\rm out}$ through the circuit in Fig.~\ref{fig:lem-s} still belongs to $\{\sigma_k\}_k$. Furthermore, the fidelity of $\eta_{\rm out}$ with the maximally entangled state $\Phi^+$ is 
  \begin{align}
      f_{\eta}=\frac{(1+a_j)(1+\gamma_{\star})}{(1+\gamma_{\star}+2a_j\gamma_{\star})},
  \end{align}
  where $\gamma_{\star}$ is the parameter of $\rho_{\star}$.
\end{lemma}
\begin{proof}
     Given an input state $ \eta_{\rm in}=\sigma_j\otimes \rho_{\star}$ with $\sigma_j=\Phi^++a_j\left(\ketbra{00}{11}+\ketbra{11}{00}\right)$, $\rho_{\star}=(2f_{\star}-1)\Phi^++2(1-f_{\star})\ketbra{00}{00}$, it is easy to obtain the output state after circuit in Fig.~\ref{fig:lem-s} is 
    \begin{align}
        \eta_{\rm out}&=\frac{1}{2}\left(\ketbra{00}{00}+\ketbra{11}{11}\right)+\left(b+\frac{1}{2}\right)\left(\ketbra{00}{11}+\ketbra{11}{00}\right)\\
        &=\Phi^++b\left(\ketbra{00}{11}+\ketbra{11}{00}\right),
    \end{align}
    with $b=\frac{a_j(1-\gamma_{\star})}{1+\gamma_{\star}+2a_j\gamma_{\star}}$,
    which satisfies $\eta_{\rm out}\in\{\sigma_k\}_k$. Moreover, the fidelity is $ f_{\eta}=\frac{(1+a_j)(1+\gamma_{\star})}{(1+\gamma_{\star}+2a_j\gamma_{\star})}$.
\end{proof}

\begin{figure}[h]
    \centering
    \includegraphics[width=0.4\linewidth]{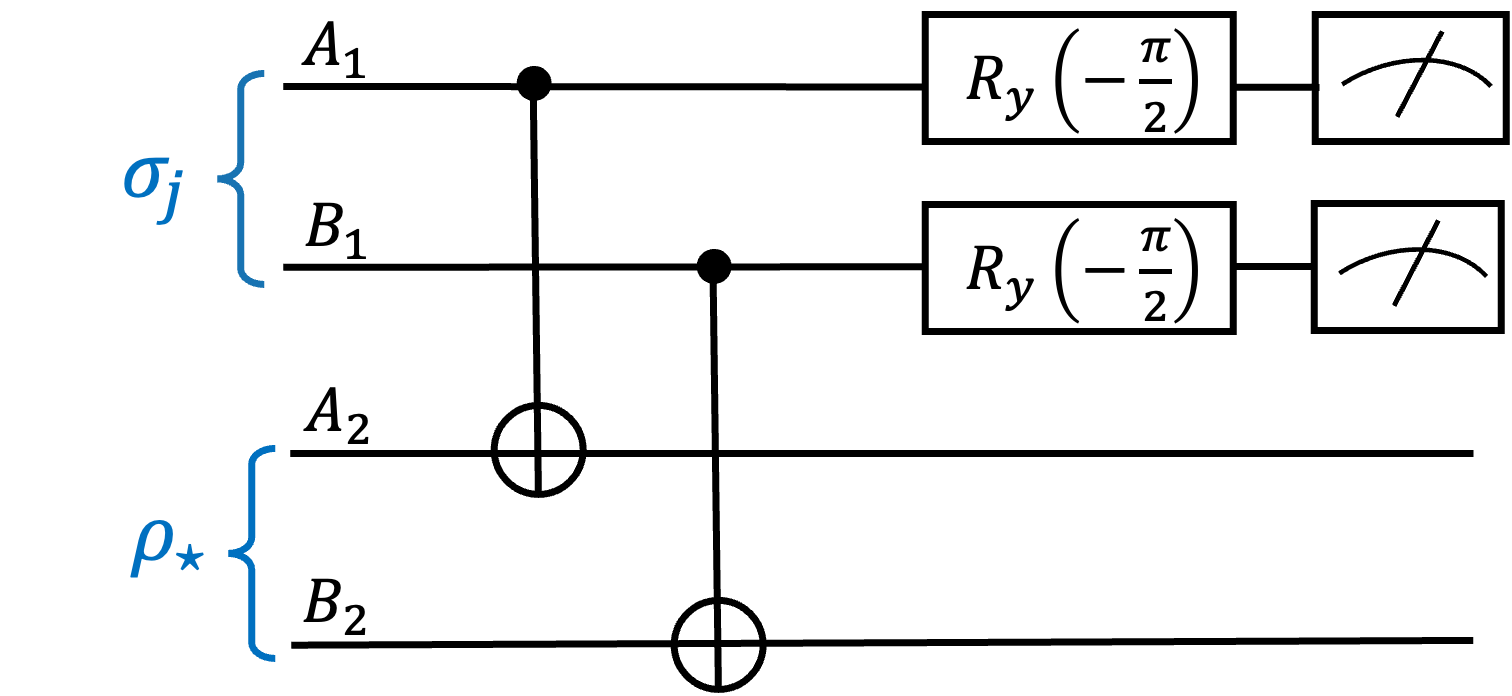}
    \caption{S state distillation protocol learned by dynamic LOCCNet.}
    \label{fig:lem-s}
\end{figure}

\begin{figure}[h]
    \centering
    \includegraphics[width=0.35\linewidth]{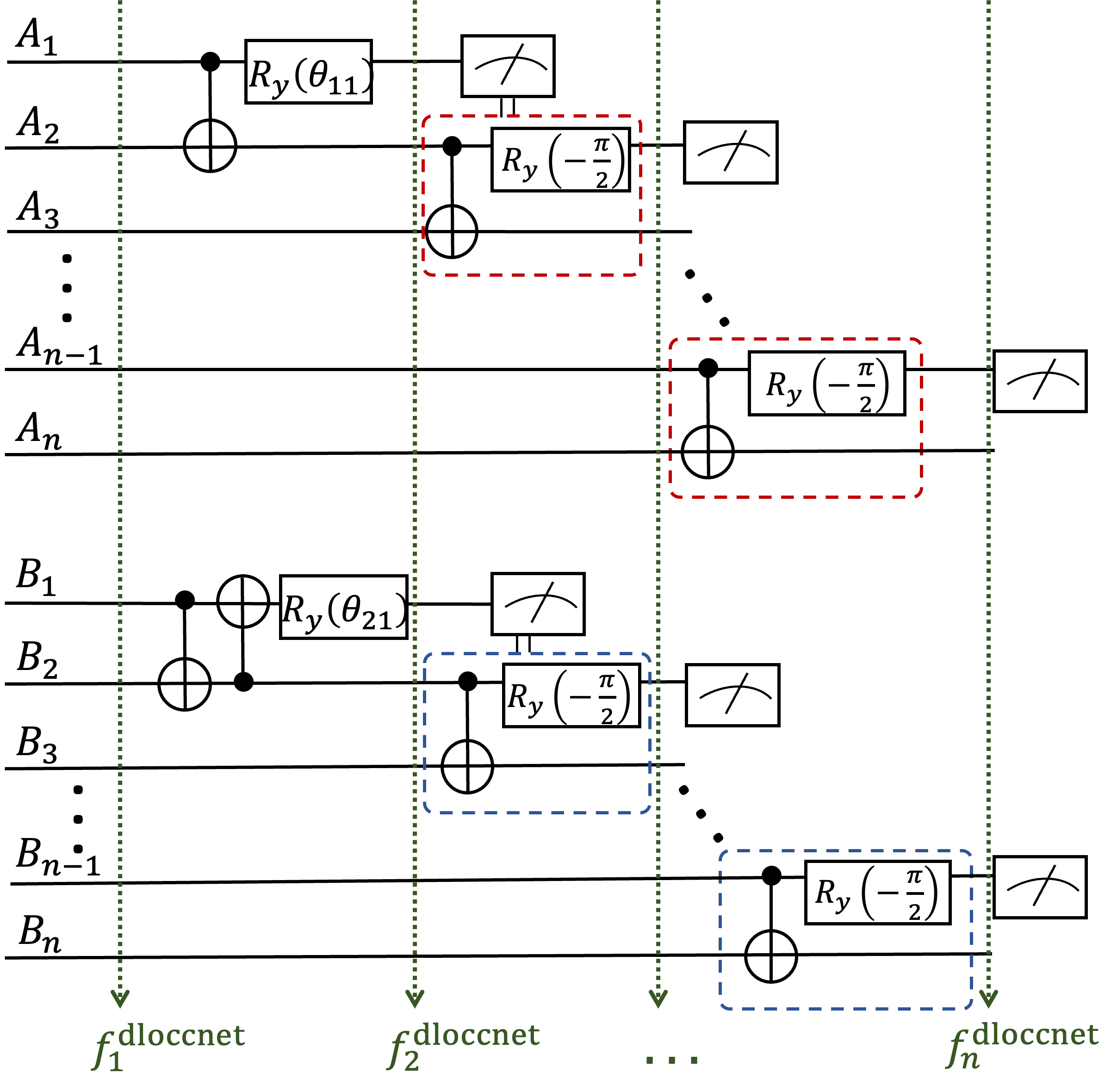}
    \caption{\textbf{Dynamic LOCCNet distillation protocol for S states.} The rotation angles of $R_y$ gates are $\theta_{11}=\theta_{21}=\arccos{(1-\gamma)+\pi}$.}
    \label{fig:propdyn-s}
\end{figure}

\begin{proposition}\label{prop:s-dyn}
    Let $\rho\in\{\rho_{\rm S}^{\gamma_i}\}_i$ be an S state defined as
  \begin{align}
      \rho=(2f_1-1)\Phi^++2(1-f_1)\ketbra{00}{00},\label{eq:s}
  \end{align}
  where $\Phi^+$ is the maximally entangled state and $f_1=\frac{1+\gamma}{2}$ with $\gamma$ being the parameter characterizing the S state. Consider the distillation protocol learned by dynamic LOCCNet as depicted in Fig.~\ref{fig:propdyn-s}, the fidelity of the output state is given by
\begin{equation}
    f_n^{\rm dloccnet}=\frac{f_1f_{n-1}}{1-f_1+(2f_1-1)f_{n-1}}, n\geq 3,
\end{equation}
where $n$ is the copy number of S state, $f_1=\frac{1+\gamma}{2}$, and 
\begin{equation}
    f_2^{\rm dloccnet}=\frac{1}{2}\left(1+\sqrt{-\gamma(\gamma-2)}\right).
\end{equation}

\end{proposition}
\begin{proof}
   Given an S state $\rho=\gamma\Phi^++(1-\gamma)\ketbra{00}{00}$, it is obvious that its fidelity to the maximally entangled state is $f_1=\frac{1+\gamma}{2}$. Then we will see that the output state after the first round is
   \begin{align}
       \sigma_{\rm out_1}=\Phi^++a_1(\ketbra{00}{11}+\ketbra{11}{00})
   \end{align}
   with $a_1=\frac{1}{2}\left(\sqrt{-\gamma(\gamma-1)}-1\right)\in[-1,0]$, corresponding to the fidelity 
   \begin{align}
     f_2^{\rm dloccnet}=\frac{1}{2}\left(1+\sqrt{-\gamma(\gamma-2)}\right). 
   \end{align}
   Note that $a_1=f_2^{\rm dloccnet}-1$, and the input state of the next round of distillation is $\sigma_{\rm out_1}\otimes\rho_S$, then we will find that the fidelity of using 3-copy of S state is 
   \begin{align}
       f_3^{\rm dloccnet}&=\frac{(1+a_1)(1+\gamma)}{1+\gamma+2a_1\gamma}\\
       &=\frac{f_1f_{2}}{1-f_1+(2f_1-1)f_{2}},
   \end{align}
   where the first equality is maintained according to the Lemma~\ref{lem:rho-s}, and the second holds due to $a_1=f_2^{\rm dloccnet}-1$ and $f_1=\frac{1+\gamma}{2}$. Similarly, when using $n$ copies of S state, the fidelity of the output state is 
\begin{equation}
    f_n^{\rm dloccnet}=\frac{f_1f_{n-1}}{1-f_1+(2f_1-1)f_{n-1}}, n\geq 3.
\end{equation}
\end{proof}

\section{Analysis of isotropic states distillation}\label{app:iso}
\subsection{Iteration method for isotropic states distillation}

\begin{lemma}\label{lem:rhoin-ori}
   Let $\rho_i\in\{\rho_{\rm iso}^{p_i}\}_i$ be an isotropic state defined as
\begin{align}
    \rho_i=f_i\Phi^++\frac{1-f_i}{3}(I-\Phi^+),
\end{align}
where $\Phi^+$ is the maximally entangled state and $f_i=\frac{1+3p_i}{4}$ with $p_i$ being the parameter characterizing the isotropic state. Consider the input state given by
\begin{align}
    \sigma_{\rm in}=\rho_j\otimes \rho_k^{\otimes 3}
\end{align}
where $\rho_j,\rho_k\in\{\rho_{\rm iso}^{p_i}\}_i$. Then the distilled state $\sigma_{\rm out}$ obtained via distillation protocol learned by LOCCNet remains an isotropic state, i.e.,
\begin{align}
    \sigma_{\rm out}\in \{\rho_{\rm iso}^{p_i}\}_i.
\end{align}
\end{lemma}
\begin{proof}
    Given an input state $ \sigma_{\rm in}=\rho_j\otimes \rho_k^{\otimes 3}$ with $\rho_j=f_j\Phi^++\frac{1-f_j}{3}(I-\Phi^+), \rho_k=f_k\Phi^++\frac{1-f_k}{3}(I-\Phi^+)$, it is easy to obtain the distilled state using distillation protocol $U_A,U_B$ in Fig.~\ref{fig:iso-16}$(a)$ is 
    \begin{align}
        \sigma_{\rm out}=f_{\sigma}\Phi^++\frac{1-f_{\sigma}}{3}(I-\Phi^+),
    \end{align}
    with 
    \begin{align}
       f_{\sigma}=\frac{1-(2+3f_j)f_k+(1+12f_j)f_k^2}{3+f_j(1-4f_k)^2-3f_k},
    \end{align}
    which satisfies $\sigma_{\rm out}\in \{\rho_{\rm iso}^{p_i}\}_i$.
\end{proof}

\begin{proposition}\label{prop:ori-itr}
Let $\rho_0\in\{\rho_{\rm iso}^{p_i}\}_i$ be an isotropic state defined as
\begin{align}
    \rho_0=f_0\Phi^++\frac{1-f_0}{3}(I-\Phi^+),
\end{align}
where $\Phi^+$ is the maximally entangled state and $f_0=\frac{1+3p_0}{4}$ with $p_0$ being the parameter characterizing the isotropic state. Consider the distillation protocol $U_A, U_B$ learned by LOCCNet as depicted in Fig.~\ref{fig:iso-16} $(a)$, the distilled state remains within the set of isotropic states after $n$ iterations of the distillation process. The fidelity of the output state after $n$ iterations, denoted as $f_{n}^{\rm ori}$, is given by the recurrence relation
\begin{align}
 f_n^{\rm ori}=\frac{1-4f_{n-1}+6f_{n-1}^2}{3-8f_{n-1}+8f_{n-1}^2},\; n\geq 1,
\end{align}
where $f_{n-1}$ is the fidelity after $n-1$ iterations.
\end{proposition}
\begin{proof}
    Given an input state $\sigma_{\rm in}=\rho_0^{\otimes 4}$ with $\rho_0\in\{\rho_{\rm iso}^{p_i}\}_i$, we find that the output state $\sigma_{\rm out}$ after the distillation protocol $U_A,U_B$ in Fig.~\ref{fig:iso-16} $(a)$ still belongs to the set $\rho_0\in\{\rho_{\rm iso}^{p_i}\}_i$ according to Lemma~\ref{lem:rhoin-ori}. Moreover, the fidelity of $\sigma_{\rm out}$ with a maximally entangled state is 
    \begin{align}
        f_1^{\rm ori}=\frac{1-4f_0+6f_0^2}{3-8f_0+8f_0^2},
    \end{align}
    where $f_0=\frac{1+3p_0}{4}$ with $\rho_0=p_0\Phi^++\frac{1-p_0}{4}I$. Since the input state of the $n$-th iteration is the output state of the $(n-1)$-th round, which remains an isotropic state, denoted as $\sigma_{\rm out}^{n-1}=f_{n-1}\Phi+\frac{1-f_{n-1}}{3}(I-\Phi^+)$, the fidelity of the output state after $n$ iterations is therefore
    \begin{align}
 f_n=\frac{1-4f_{n-1}+6f_{n-1}^2}{3-8f_{n-1}+8f_{n-1}^2},\; n\geq 1.
\end{align}
\end{proof}

\subsection{Dynamic method for isotropic states distillation}
\begin{proposition}\label{prop:dyn-itr}
Let $\rho_0\in\{\rho_{\rm iso}^{p_i}\}_i$ be an isotropic state defined as
\begin{align}
    \rho_0=f_0\Phi^++\frac{1-f_0}{3}(I-\Phi^+),
\end{align}
where $\Phi^+$ is the maximally entangled state and $f_0=\frac{1+3p_0}{4}$ with $p_0$ being the parameter characterizing the isotropic state.Consider the distillation protocol $U_A, U_B$ learned by LOCCNet as depicted in Fig.~\ref{fig:iso-16} $(a)$, the distilled state remains within the set of isotropic states after $n$ iterations using the method of dynamic LOCCNet. The fidelity of the output state after $n$ iterations, denoted as $f_{n}^{\rm dyn}$, is given by the recurrence relation
\begin{align}
  f_n^{\rm dyn}=\frac{1-2f_0+f_0^2-3f_0f_{n-1}+12f_0^2f_{n-1}}{3-3f_0+f_{n-1}-8f_0f_{n-1}+16f_0^2f_{n-1}},\; n\geq 2,
\end{align}
where $f_{n-1}$ is the fidelity after $n-1$ iterations with $f_1=\frac{1-2p_0+9p_0^2}{4-8p_0+12p_0^2}$.
\end{proposition}

\begin{proof}
    Given an input state $\sigma_{\rm in}=\rho_0'\otimes\rho_0^{\otimes 3}$ with $\rho_0',\rho_0\in\{\rho_{\rm iso}^{p_i}\}_i$, we find that the output state $\sigma_{\rm out}$ after the distillation protocol $U_A,U_B$ in Fig.~\ref{fig:iso-16} $(a)$ still belongs to the set $\rho_0\in\{\rho_{\rm iso}^{p_i}\}_i$ according to Lemma~\ref{lem:rhoin-ori}. Moreover, the fidelity of $\sigma_{\rm out}$ after first iteration using dynamic method in Fig.~\ref{fig:iso-16} $(b)$ is $f_1^{\rm dyn}=\frac{1-4f_0+6f_0^2}{3-8f_0+8f_0^2}$, where $f_0=\frac{1+3p_0}{4}$ with $\rho_0=p_0\Phi^++\frac{1-p_0}{4}I$. Since the input state of the $n$-th iteration is $\sigma_{\rm in}^{n}=\sigma_{\rm out}^{n-1}\otimes\rho_0^{\otimes 3}$, with the $(n-1)$-th round output state $\sigma_{n-1}=f_{n-1}\Phi+\frac{1-f_{n-1}}{3}(I-\Phi^+)$, which remains an isotropic state, the fidelity of the output state after $n$ iterations is therefore  
    \begin{align}
  f_n=\frac{1-2f_0+f_0^2-3f_0f_{n-1}+12f_0^2f_{n-1}}{3-3f_0+f_{n-1}-8f_0f_{n-1}+16f_0^2f_{n-1}},\; n\geq 2,
\end{align}
with $f_1=\frac{1-2p_0+9p_0^2}{4-8p_0+12p_0^2}$.
\end{proof}

\section{Other cases of distillation and discrimination}\label{app:other-distill-discri}
We extend our analysis to the distillation of maximally entangled states in qutrit case that are affected by depolarizing channel. The maximally entangled states in 3-dimension is 
\begin{align}
    \ket{\Psi^+}_{\rm qutrit}=\frac{1}{\sqrt{3}}(\ket{00}+\ket{11}+\ket{22}).
\end{align}
When each qutrit of $\ket{\Psi^+}_{\rm qutrit}$ passes through the depolarizing channel $\cN_{\rm dep}$, the resulting mixed state is
\begin{align}
    \sigma=p\rho+(1-p)\frac{I}{9},
\end{align}
where $\rho=\ket{\Psi^+}_{\rm qutrit}\langle \Psi^{+}|$.

In our analysis, we employ DLOCCNet with 2-in-1-out to distill the noisy state $\sigma$ using three copies. The final fidelity achieved for this scenario is presented in Fig.~\ref{fig:distill-qutrit}. The results demonstrate that DLOCCNet consistently effectively distills the noisy states across various noise parameters, indicating that our framework is also applicable and effective for the high-dimensional system distillation.
\begin{figure}[h]
    \centering
    \includegraphics[width=0.6\linewidth]{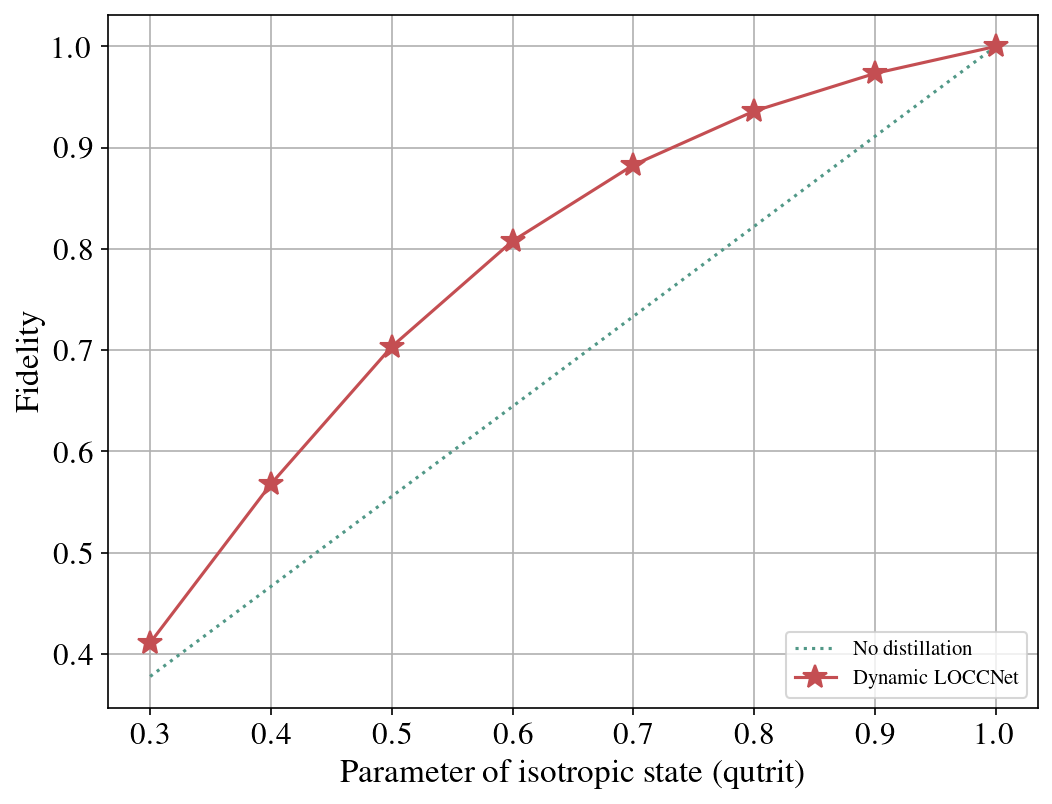}
    \caption{Fidelity achieved by distillation protocols for qutrit case.}
    \label{fig:distill-qutrit}
\end{figure}

\begin{figure}[h]
    \centering
    \includegraphics[width=0.6\linewidth]{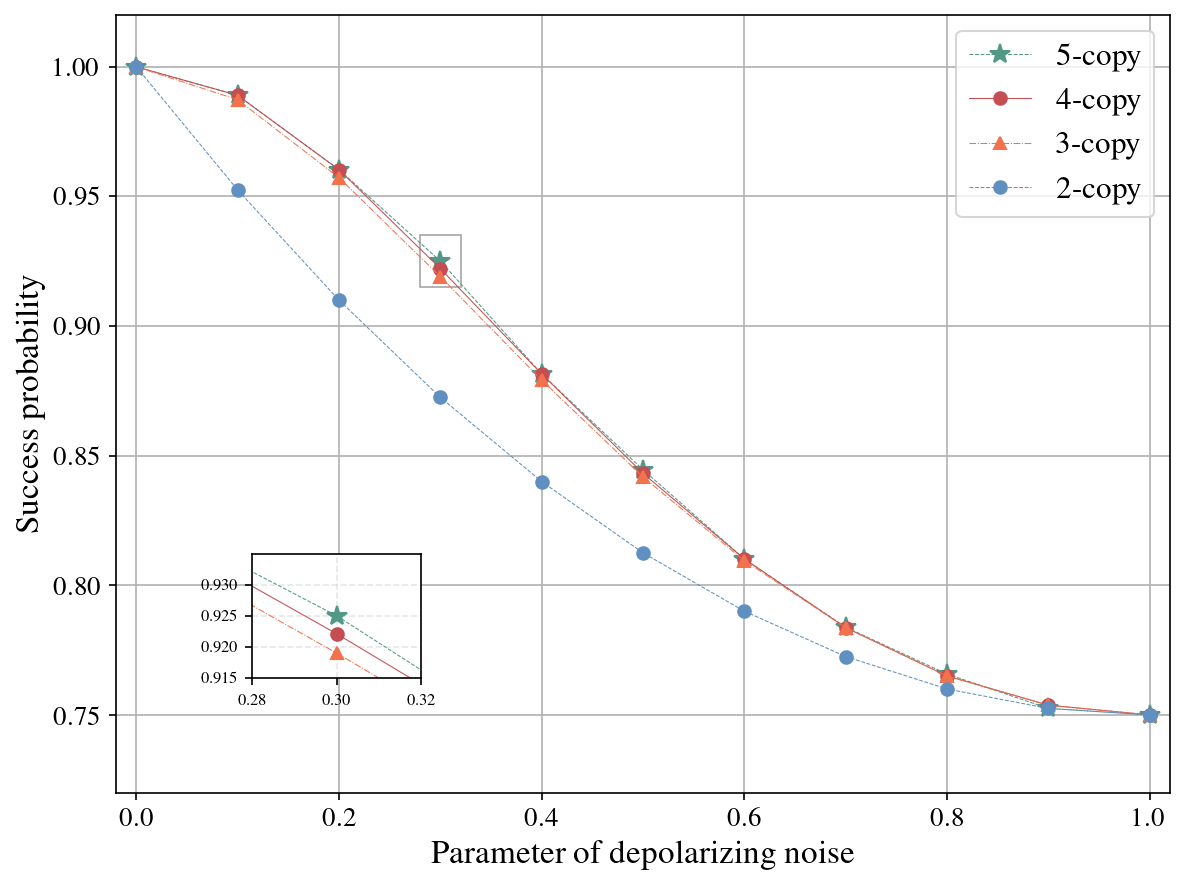}
    \caption{Average success probability of distinguishing a Bell state and a noisy Bell state affected by depolarizing noise using different copy numbers.}
    \label{fig:distill-qutrit}
\end{figure}

\end{document}